\documentclass{vgtc}                          

\newif\ifwarntodo 
\warntodotrue

\newif\ifshowtodo
\showtodotrue

\usepackage{times}                     

\usepackage{tabu}                      
\usepackage{booktabs}                  
\usepackage{lipsum}                    
\usepackage{mwe}                       

\usepackage{mathptmx}                  

\graphicspath{{pdfimages/}, {images/}} 
\usepackage{subcaption}   

\usepackage{booktabs}
\usepackage{multirow}
\usepackage{microtype}

\usepackage[svgnames]{xcolor}

\usepackage{makecell}   

\usepackage{tikz}
\usetikzlibrary{patterns,shapes.arrows}
\usetikzlibrary{positioning}
\usetikzlibrary{arrows.meta}

\usetikzlibrary{external}
\tikzset{external/only named=true}

\usepackage{tikzscale}

\usepackage{pgfplots}
\usepackage{pgfplotstable}
\usepgfplotslibrary{groupplots}
\pgfplotsset{compat=newest}

\usepackage{amsmath}
\usepackage{amsbsy}
\usepackage{amsfonts}
\usepackage{mathtools}
\usepackage{amssymb, amsthm}  
\usepackage{nicefrac}

\usepackage{cite}                      

\onlineid{1034}

\vgtccategory{Research}

\vgtcinsertpkg

\title{Revisiting Accurate Geometry for Morse-Smale Complexes}

\author{Son Le Thanh\thanks{e-mail: sonlt@kth.se}\\ %
        \scriptsize KTH Royal Institute of Technology %
\and Michael Ankele\thanks{e-mail: michael.ankele@indurad.com}\\ %
        \scriptsize Indurad %
\and Tino Weinkauf\thanks{e-mail: weinkauf@kth.se}\\ %
        \scriptsize KTH Royal Institute of Technology}

\abstract{
    The Morse-Smale complex is a standard tool in visual data analysis. The classic definition is based on a 
continuous view of the gradient of a scalar function where its zeros are the critical points. These points are connected via 
gradient curves and surfaces emanating from saddle points, known as separatrices. In a discrete setting, the Morse-Smale complex 
is commonly extracted by constructing a combinatorial gradient assuming the steepest descent direction.
Previous works have shown that this method results in a geometric embedding of the separatrices that can be fundamentally different from
those in the continuous case. To achieve a similar embedding, different approaches for constructing a combinatorial gradient
were proposed. 
In this paper, we show that these approaches generate a different topology, i.e., the connectivity between critical points changes.
Additionally, we demonstrate that the steepest descent method can compute topologically and geometrically accurate Morse-Smale 
complexes when applied to certain types of grids. Based on these observations, we suggest a method to attain both geometric and 
topological accuracy for the Morse-Smale complex of data sampled on a uniform grid.
} 


\begin{document}



\newcommand{\RRR}{{\mathrm I\! \mbox{R} }}
\newcommand{\EEE}{{\mathrm I\! \mbox{E} }}
\newcommand{\xx}{{\mathbf x}}
\newcommand{\yy}{{\mathbf y}}
\newcommand{\zz}{{\mathbf z}}
\newcommand{\dd}{{\mathbf d}}
\newcommand{\hh}{{\mathbf h}}
\newcommand{\ttt}{{\mathbf t}}
\newcommand{\jj}{{\mathbf j}}
\newcommand{\pp}{{\mathbf p}}
\newcommand{\qq}{{\mathbf q}}
\newcommand{\aaa}{{\mathbf a}}
\newcommand{\ssss}{{\mathbf s}}
\newcommand{\bb}{{\mathbf b}}
\newcommand{\ee}{{\mathbf e}}
\newcommand{\cc}{{\mathbf c}}
\newcommand{\nn}{{\mathbf n}}
\newcommand{\mm}{{\mathbf m}}
\newcommand{\kk}{{\mathbf k}}
\newcommand{\rr}{{\mathbf r}}
\newcommand{\uu}{{\mathbf u}}
\newcommand{\vv}{{\mathbf v}}
\newcommand{\ff}{{\mathbf f}}
\newcommand{\fff}{{\mathbf f}}
\newcommand{\JJ}{{\mathbf J}}
\newcommand{\DD}{{\mathbf D}}
\newcommand{\BB}{{\mathbf B}}
\newcommand{\CC}{{\mathbf C}}
\newcommand{\EE}{{\mathbf E}}
\newcommand{\PP}{{\mathbf P}}
\newcommand{\MM}{{\mathbf M}}
\newcommand{\LL}{{\mathbf L}}
\newcommand{\QQ}{{\mathbf Q}}
\newcommand{\FF}{{\mathbf F}}
\newcommand{\TT}{{\mathbf T}}
\newcommand{\RR}{{\mathbf R}}
\newcommand{\SSS}{{\mathbf S}}
\newcommand{\Sa}{{\mathbf {Sa}}}
\newcommand{\ZZ}{Z}
\newcommand{\cs}{{\dot c}}
\newcommand{\sss}{{\dot s}}
\newcommand{\ps}{{\dot p}}
\newcommand{\pg}{{\dot g}}
\newcommand{\ww}{{\mathbf w}}
\newcommand{\xxs}{\dot{\mathbf x}}
\newcommand{\vvs}{\dot{\mathbf v}}
\newcommand{\us}{{\dot u}}
\newcommand{\vs}{{\dot v}}
\newcommand{\ws}{{\dot w}}
\newcommand{\fn}{{\mathbf 0}}

\newcommand{\ii}{{\mathbf i}}
\newcommand{\Dt}{{\widetilde D}}
\newcommand{\fft}{{\widetilde \ff}}
\newcommand{\kkt}{{\widetilde \kk}}
\newcommand{\ggt}{{\widetilde \ggg}}
\newcommand{\hht}{{\widetilde \hh}}
\newcommand{\vvt}{{\widetilde \vv}}
\newcommand{\wwt}{{\widetilde \ww}}

\newcommand{\dbar}[1]{\bar{\bar{#1}}}
\newcommand{\qqqq}{\dbar{\qq}}

\newcommand{\eet}{\dbar{\ee}_{t}}

\newcommand{\eetau}{\dbar{\ee}_{\tau}}
%

\newcommand{\mNN}{\mathcal N}
\newcommand{\mEE}{\mathcal E}
\newcommand{\mSS}{\mathcal S}
\newcommand{\mRR}{\mathcal R}
\newcommand{\mPP}{\mathcal P}
\newcommand{\mGG}{\mathcal G}
\newcommand{\mHH}{\mathcal H}
\newcommand{\mVV}{\mathcal V}
\newcommand{\mTT}{\mathcal T}
\newcommand{\mLL}{\mathcal L}
\newcommand{\mBB}{\mathcal B}

\newcommand{\UF}{\texttt{UF}\ }

\newtheorem{theorem}{Theorem}
\newtheorem{lemma}{Lemma}
\newtheorem{corollary}{Corollary}
\newcommand{\magn}[1]{|#1|}
\newcommand{\cardinality}[1]{|#1|}
\newcommand{\norm}[1]{\left\lVert#1\right\rVert}

%
%




\ifwarntodo
	\newcommand{\warn}{\PackageWarning{}{Unprocessed note}}
\else
	\newcommand{\warn}{}
\fi

\newcommand{\todo}[1]{\warn%
\ifshowtodo%
\textbf{\small\textcolor{red}{TODO: #1}}%
\fi%
}

\newcommand{\annoauthor}[3]{\warn%
\ifshowtodo%
\textbf{\small\textcolor[rgb]{#1}{#2: \emph{#3}}}%
\fi%
}

\newcommand{\son}[1]{\annoauthor{0.2,0.2,0.8}{Son}{#1}}
\newcommand{\tino}[1]{\annoauthor{0.8,0.2,0.2}{Tino}{#1}}

\newcommand{\changedCamera}[1]{#1}
\newcommand{\changed}[1]{#1}
\newcommand{\allchanged}{}
\newcommand{\alladded}{}
\newcommand{\allsame}{}
\newcommand{\added}[1]{#1}

\newlength{\lengthgoodgap}
\addtolength{\lengthgoodgap}{12pt}
\newcommand{\goodgap}{\hspace{\lengthgoodgap}}

\newcommand{\goodnewline}{\\[0.5\lengthgoodgap]}

\newlength{\twopicwidth}
\addtolength{\twopicwidth}{0.5\textwidth}
\addtolength{\twopicwidth}{-0.5\lengthgoodgap}

\newlength{\threepicwidth}
\addtolength{\threepicwidth}{0.333333\textwidth}
\addtolength{\threepicwidth}{-0.666666\lengthgoodgap}

\newlength{\fourpicwidth}
\addtolength{\fourpicwidth}{0.25\textwidth}
\addtolength{\fourpicwidth}{-0.75\lengthgoodgap}

\newlength{\fivepicwidth}
\addtolength{\fivepicwidth}{0.20\textwidth}
\addtolength{\fivepicwidth}{-0.80\lengthgoodgap}

\newlength{\sixpicwidth}
\addtolength{\sixpicwidth}{0.166666666666666667\textwidth}
\addtolength{\sixpicwidth}{-0.833333333333333333\lengthgoodgap}

\newlength{\ltwopicwidth}
\addtolength{\ltwopicwidth}{0.5\linewidth}
\addtolength{\ltwopicwidth}{-0.5\lengthgoodgap}

\newlength{\lthreepicwidth}
\addtolength{\lthreepicwidth}{0.333333\linewidth}
\addtolength{\lthreepicwidth}{-0.666666\lengthgoodgap}

\newlength{\lfourpicwidth}
\addtolength{\lfourpicwidth}{0.25\linewidth}
\addtolength{\lfourpicwidth}{-0.75\lengthgoodgap}

\hfuzz=3pt

\def\chapterautorefname{Chapter}
\def\sectionautorefname{Section}
\def\subsectionautorefname{Section}
\def\figureautorefname{Figure}
\def\subfigureautorefname{Figure}
\def\tableautorefname{Table}
\def\equationautorefname~#1\null{(#1)\null}

\newcommand{\Autoref}[1]{%
  \begingroup%
  \def\chapterautorefname{Chapter}%
  \def\sectionautorefname{Section}%
  \def\subsectionautorefname{Section}%
  \def\figureautorefname{Figure}%
	\def\subfigureautorefname{Figure}
  \def\tableautorefname{Table}%
  \def\equationautorefname~##1\null{Equation~(##1)\null}%
  \autoref{#1}%
  \endgroup%
}

\newcommand{\SupplementalMaterialHeadWithNames}[1]{%
\begin{center}
Supplemental Material -- #1
\end{center}

\vspace{-1.5\baselineskip}
\section*{\centering Temporal Merge Tree Maps:\\A Topology-Based Static Visualization for Temporal Scalar Data}
\subsection*{\centering Wiebke Köpp and Tino Weinkauf}
}

\newcommand{\resetlength}[2]{\ifx#1\undefined \newlength{#1} \fi \setlength{#1}{#2} }

%
%





\maketitle

\section{Introduction} 
\firstsection{Introduction}

\maketitle


The Morse-Smale complex \cite{morse34,smale1961gradient}
has proven to be a powerful tool
for topological data analysis
and has found its applications
in many fields
such as in material science \cite{gyulassy07distancefields},
fluid dynamics \cite{sahner07a, guenther12b},
computer graphics \cite{weinkauf09b},
or in molecular biology \cite{cazals03}.

Forman's discrete Morse theory \cite{forman02}
gives a formidable setting
for computing the Morse-Smale complex
in a combinatorial setting \cite{robins10,guenther12a,gyulassy14conforming}.
The essential component of
these algorithms is the computation of the discrete gradient field,
which implicitly encodes the structure of the Morse-Smale complex.
The remarkably efficient method by Robins et al.\ \cite{robins10}
can be considered the \emph{de facto} standard
for the calculation of the discrete gradient field.
The method aligns its discrete vectors locally with the steepest descent direction.

It has been shown
by Gyulassy et al.\ \cite{gyulassy12ms} and Reininghaus et al.\ \cite{reininghaus2012combinatorial}
that the geometric embedding
of the discrete Morse-Smale complex
does not coincide with its continuous counterpart
when using the steepest descent direction in Robins' method
-- even when increasing the resolution of the discretization.
This has been investigated for uniform grids,
but not for other types of grids.
Instead,
Gyulassy et al.\ \cite{gyulassy12ms} and Reininghaus et al.\ \cite{reininghaus2012combinatorial}
propose a new method
to compute the discrete gradient field,
which replaces the steepest descent direction in Robins et al.\ \cite{robins10}
with a probabilistic choice of descending directions.
This achieves a better geometric embedding
of the Morse-Smale complex.

However,
as we will show in this paper,
moving in directions other than the steepest descent
will lead to a different topology.
Specifically,
the connectivity between critical points
differs.
We show
that these differences
can be found in many scenarios,
from small data sets to large ones,
from smooth data sets to noisy ones,
from artificial data sets to real-world data sets.
We provide a theoretical discussion
that reveals which layers of the cell complex
are affected
by the probabilistic approach
of Gyulassy et al.\ \cite{gyulassy12ms} and Reininghaus et al.\ \cite{reininghaus2012combinatorial}.

To provide solutions,
we investigate different types of grids
for their susceptibility to distortions of the geometric embedding.
Based on these observations,
we propose a method
to convert a uniform grid
to a specific triangle grid
that provides as good of a geometric embedding as
the probabilistic approaches
of Gyulassy et al.\ \cite{gyulassy12ms} and Reininghaus et al.\ \cite{reininghaus2012combinatorial},
but can be used with the steepest descent method of Robins et al.\ \cite{robins10},
thereby inheriting its topological guarantees.

Our contributions are as follows:
\begin{itemize}

\item
\changed{We provide a novel, systematic analysis of the 
differences in the topology of Morse-Smale complexes
produced by various methods}
(\Cref{section-TopologicalInconsistencies}).
These changes happen across a wide range of scenarios 
and can not be removed by the means of topological simplification. 

\item
We investigate different types of grids
and identify \changed{empirically} the properties
that a grid needs to have
such that the geometric embedding obtained using the steepest descent method
aligns with the continuous case (\Cref{section-EffectOfSampling}).

\item
We propose a method to achieve both geometric accuracy and topological consistency
for data sampled on a uniform grid
by converting it to a specific type of triangle grid
and applying the steepest descent method (\Cref{section-Suggestion}). 

\end{itemize}

Related work and theoretical background
are discussed in the following two sections,
with a focus on the topic of accurate geometry
in \Cref{section-AccurateGeometryForMSComplexes}.

\section{Background}
\label{section-preliminaries}

We briefly go over the relevant concepts in the following.

\subsection{Morse function and Morse-Smale complex}
\label{section-MS Complex}

Morse theory \cite{milnor63} allows to analyze the topology of a manifold by examining a function defined on that manifold.
Let $f: \MM \rightarrow \RR$ be a function defined on a $d$-manifold with boundary $\MM$, $\pp \in \MM$ is a \emph{critical point} of $f$ 
if $\nabla f(\pp) = 0$. If the Hessian of $f$ at $\pp$, $H(\pp)$, is non-singular, i.e. $|H(\pp)| \neq 0$, then $\pp$ is called a
\emph{non-degenerate critical point}.
The function $f$ is called a \emph{Morse function} if all of its critical points are non-degenerate.
Given a Morse function $f$ with a non-degenerate critical point $\pp$,
the result of the Morse lemma states there exists a local coordinate in the
neighborhood of $\pp$ such that $f$ can take the quadratic form
$f(\xx) = f(\pp) - x_1^2 - \ldots - x_\gamma^2 +x_{\gamma + 1}^2 + \ldots + x_d^2$.
The value of $\gamma$ in this formula is called the \emph{index} of $f$ at $\pp$ and can be used to characterize the type of the critical points of $f$.
In the case of $2$ dimensions,
the indices $0$, $1$, and $2$
correspond to the minima, saddle points, and maxima of the function, respectively.
For volumetric data,
minima have index $0$,
$1$-saddles have index $1$,
$2$-saddles have index $2$,
and maxima have index $3$.

The tangent vector
of a \emph{tangent curve}
$\phi: \RR \rightarrow \MM$
agrees with the gradient $\nabla f$
at every point along the line.
In other words,
$\phi$ is the solution of the equation
$\dfrac{\partial}{\partial t}\phi(t) = \nabla f(\phi(t))$,
where $\phi(0) = \pp$.
For each tangent curve,
the point $\underset{t \rightarrow - \infty}{\text{lim}} \phi(t)$
is called
the \emph{source} or \emph{origin} of $\phi$,
whereas $\underset{t \rightarrow  \infty}{\text{lim}} \phi(t)$
is called the \emph{destination} or \emph{sink}.
The sets of tangent curves
having the same source or sink
are called \emph{ascending} and \emph{descending} manifolds, respectively. 

A function $f$ fulfills the \emph{Morse-Smale condition}
if its ascending and descending manifolds
intersect transversally for each pair of critical points. 
If $f$ satisfies such condition,
the intersection of the ascending and descending manifolds
defines a structure known as the \emph{Morse-Smale complex}.
It segments $\MM$ into regions with monotone gradient flow behaviors,
i.e.,
all gradients inside a Morse-Smale cell
share the same source and destination.
The boundaries of the Morse-Smale regions
are called \emph{separatrices},
the unique tangent curve connecting a pair of 
critical points with consecutive indices.

\subsection{Discrete Morse theory}
\label{section-DiscreteMorseTheory}

Discrete Morse theory
developed by Forman \cite{forman02}
aims to describe
the concepts and properties
of the original Morse theory
in a combinatorial fashion.
It is widely adopted
as the foundation
for efficient extraction schema
for the Morse-Smale complex.
We review only the necessary 
concepts that will be mentioned later on.

A $d$-cell is a topological space that is homeomorphic
to the closed $d$-ball $\BB^d = \{\xx \in \EE^d |  \norm{\xx} \leq 1 \}$,
where $\EE^d$ denotes the $d$-dimensional Euclidean space.
For instance, a vertex is a $0$-cell,
an edge that connects two vertices is a $1$-cell,
and a polygon of any shape is a $2$-cell.
Given a $d$-cell $\alpha^{(d)}$,
we will write this cell as $\alpha$
if the dimension is not necessary for the context.
A cell $\alpha^{(d_1)}$ is called a \emph{face} of another cell $\beta^{(d_2)}$, denoted by $\alpha^{(d_1)} \preceq \beta^{(d_2)}$,
if $d_1 \leq d_2$
and the $0$-cells of $\alpha^{(d_1)}$
are a subset
of the $0$-cells of $\beta^{(d_2)}$.
In this case,
we also say that
$\beta^{(d_2)}$ is a \emph{co-face} of $\alpha^{(d_1)}$. 

A \emph{cell complex} $K$
is a collection of cells
such that the intersection
between any two cells is either empty or a common face of both.
To give some examples,
uniform and curvilinear grids
as well as
triangle and tetrahedral meshes
are cell complexes.
The $d$-skeleton of $K$ is the collection
of all cells whose dimension is not larger than $d$.
Let $f: K \rightarrow \RR$ be a function defined on a cell complex $K$. For any face $\alpha \in K$, 
the \emph{lower star} of $\alpha$, denoted by $\text{St}^-(\alpha) = \{\sigma \in K| \alpha \preceq \sigma \wedge f(\sigma) \leq  f(\alpha)\}$, is the
collection of \changed{co-}faces of $\alpha$ whose values are not exceeding that of $\alpha$.

Let $f: K \rightarrow \RR$ be a function
that assigns a value to every cell of $K$.
If for every face
$\alpha^{(d)} \in K$
the following holds
\begin{align}
\left| \left\{ \beta^{(d+1)}| \alpha^{(d)}  \preceq \beta^{(d+1)} \wedge f(\alpha) \geq f(\beta) \right\} \right| &\leq 1,\\
\left| \left\{ \gamma^{(d-1)} | \gamma^{(d-1)} \preceq \alpha^{(d)} \wedge f(\gamma) \geq f(\alpha) \right\} \right| &\leq 1,
\end{align}
then $f$ is called a \emph{discrete Morse function}.
If a cell $\alpha^{(d)}$ satisfies
\begin{align}
\left| \left\{ \beta^{(d+1)}| \alpha^{(d)} \preceq \beta^{(d+1)} \wedge f(\alpha) \geq f(\beta) \right\} \right| &= 0,\\
\left| \left\{ \gamma^{(d-1)} | \gamma^{(d-1)} \preceq \alpha^{(d)} \wedge f(\gamma) \geq f(\alpha) \right\} \right| &= 0,
\end{align}
then $\alpha^{(d)}$ is a critical cell with the index $d$.
A \emph{discrete vector} is a pair $\{\alpha^{(d)}, \beta^{(d+1)}\}$
such that $\alpha \preceq \beta$.
The collection $V$ of discrete vectors
where each cell of $K$ appears in at most in one vector
is called a \emph{discrete vector field}.
Given a discrete vector field $V$,
a \emph{$V$-path} is an alternating sequence of cells
\begin{align}
  \alpha_0^{(d)}, \beta_0^{(d+1)}, \alpha_1^{(d)}, \beta_1^{(d+1)}, \alpha_2^{(d)}, \ldots, \beta_\ell^{(d+1)}, \alpha_{\ell+1}^{(d)} 
\end{align}
such that $\{\alpha, \beta \} \in V$ and $\alpha_i \neq \alpha_{i+1} \preceq \beta_i$.
If $\ell \geq 0$ and $\alpha_0 = \alpha_{\ell+1}$ then the path is a 
\emph{non-trival closed path}.
If the discrete vector field $V$ does not contain any such closed path,
then $V$ is called a \emph{discrete gradient field} of $f$.
This is analogous to the continuous gradient of a function
defined on $K$.
Thus, the $V$-paths are comparable
to the tangent curves in the original Morse theory.
A $V$-path of a discrete gradient field
starts or terminates at critical cells.
Hence, we can define the \emph{discrete Morse-Smale complex} similarly 
to the continuous counterpart.

Using discrete Morse theory
provides a framework to design algorithms
to extract the Morse-Smale complex
with higher robustness and reliability.
The combinatorial nature
of this theory
enables handling degenerate cases
such as plateau regions that can appear when dealing with data.

\subsection{Topological Simplification}
\label{section-Simplification}

The goal of topological simplification
is to remove small-scale, spurious features,
thereby enabling a meaningful analysis of the data.
In the context of Forman's discrete Morse theory,
simplification of a discrete Morse function
can be done by removing a pair of connected critical points
and updating the connections of the neighbors of these two critical points.
Alternatively,
Forman \cite{forman02} suggests simplification by reversing the gradient vector field
along the path connecting two critical cells.
This approach implicitly
cancels the two mentioned critical cells
while resulting in a new valid discrete Morse function.
One can use the \emph{height difference} between two critical cells,
which is the absolute difference in the value of these cells,
to guide the simplification.
The use of the height difference is closely related to the concept of 
\emph{persistence} \cite{edelsbrunner02}.
See Günther et al.~\cite{guenther14d} for a detailed discussion.

\section{Accurate Geometry for Morse-Smale Complexes}
\label{section-AccurateGeometryForMSComplexes}

\newlength{\anafunconeimgwidth}
\setlength{\anafunconeimgwidth}{0.32\textwidth}

\begin{figure*}[t]%
\begin{subfigure}[t]{0.985\anafunconeimgwidth}%
\includegraphics[width=\linewidth]{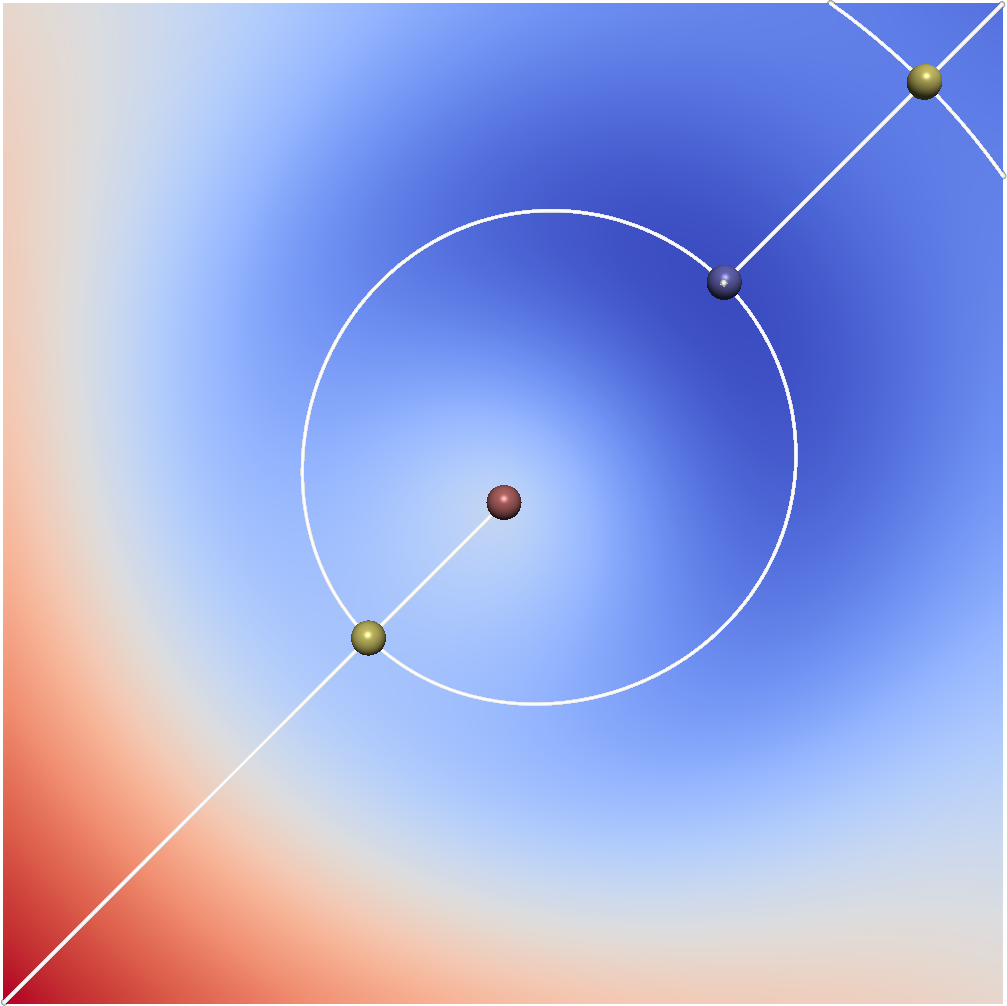}%
\caption{Continuous topology.}%
\label{fig-AnaFunc1Cont}%
\end{subfigure}%
\hfill%
\begin{subfigure}[t]{\anafunconeimgwidth}%
\includegraphics[width=\linewidth]{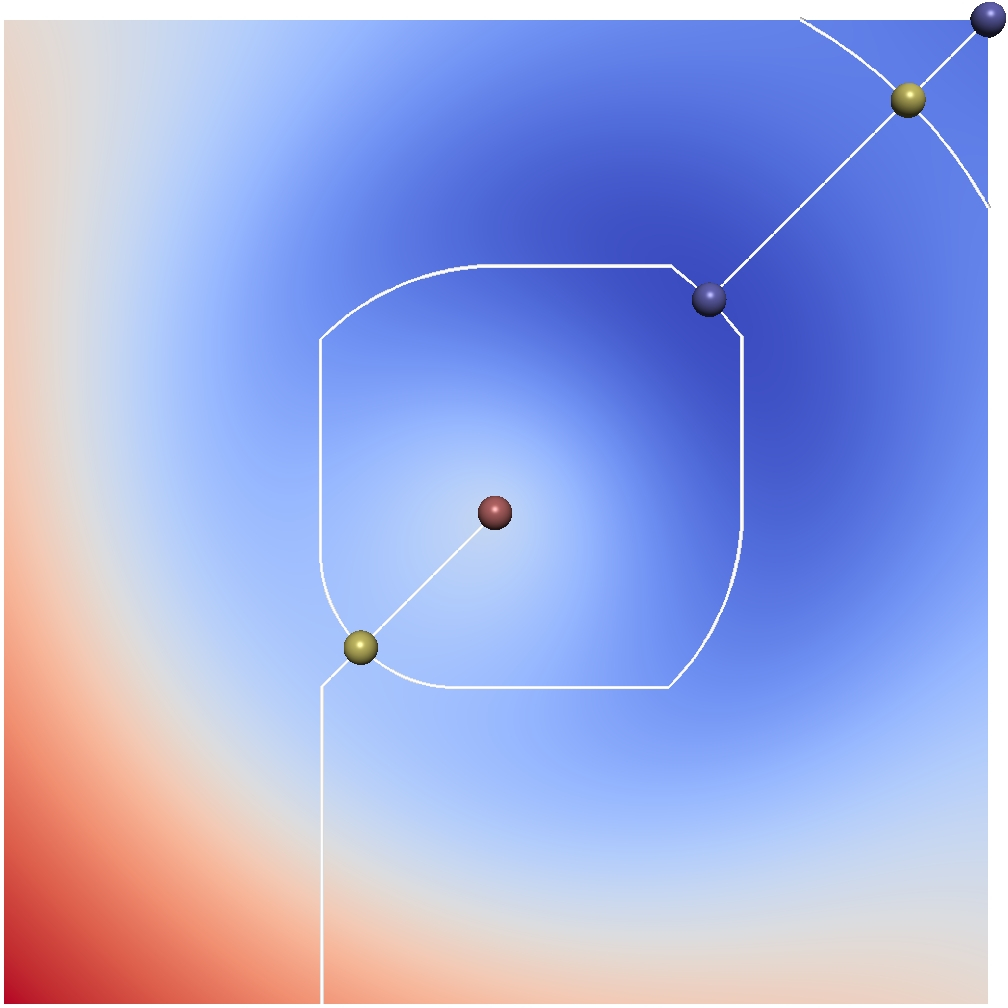}%
\caption{Steepest descent method \cite{robins10} on the uniform grid.}%
\label{fig-AnaFunc1Robins}%
\end{subfigure}%
\hfill%
\begin{subfigure}[t]{\anafunconeimgwidth}%
\includegraphics[width=\linewidth]{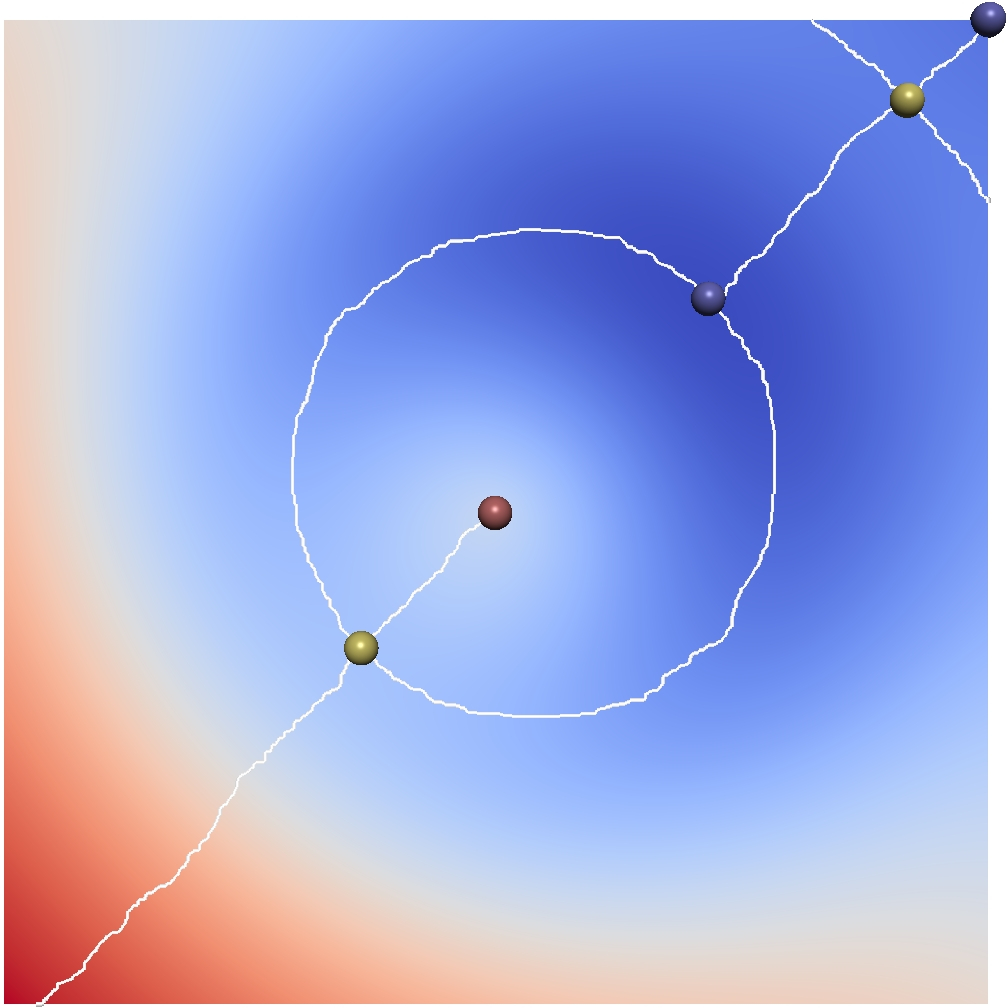}%
\caption{Probabilistic method \cite{gyulassy12ms, reininghaus2012combinatorial} on the uniform grid.}%
\label{fig-AnaFunc1Jans}%
\end{subfigure}%
\\%
\begin{subfigure}[t]{0.985\anafunconeimgwidth}%
\includegraphics[width=\linewidth]{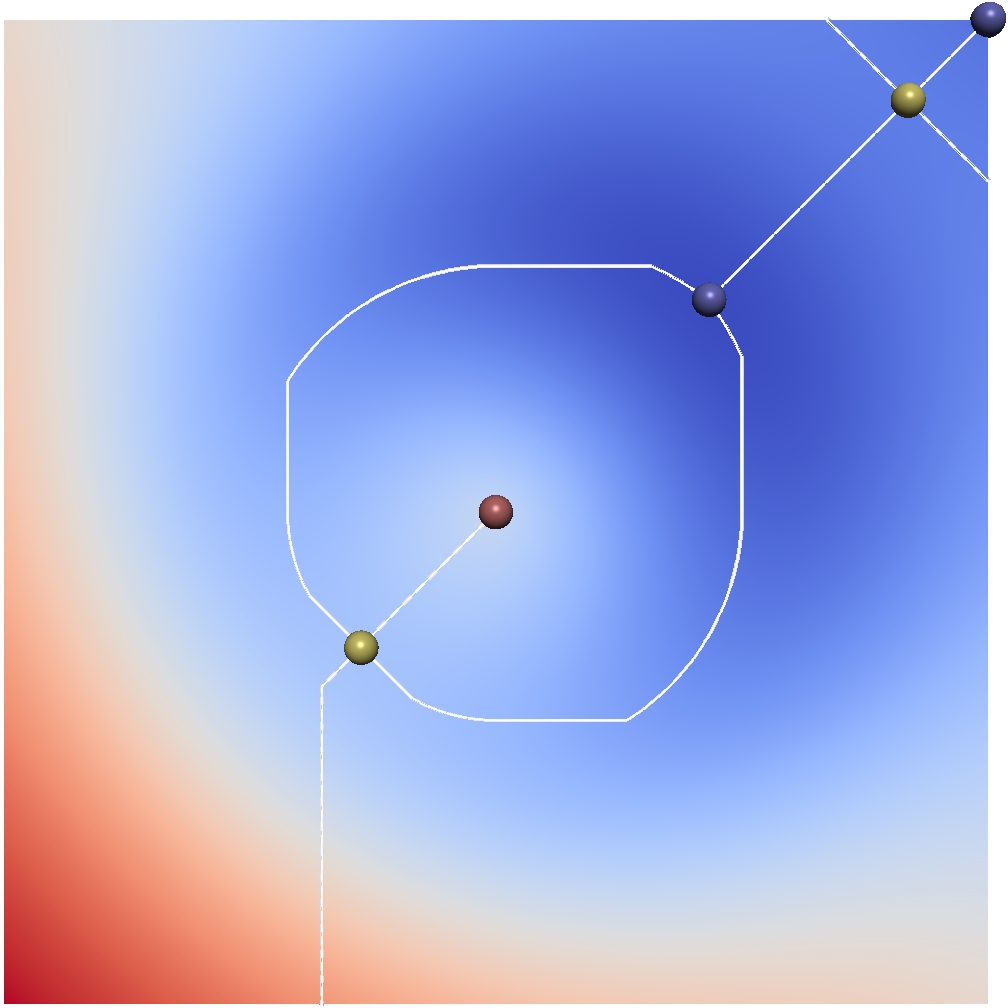}%
\caption{Steepest descent method \cite{robins10} on a triangle grid obtained by diagonally dividing each cell of a uniform grid once as shown in \Cref{fig-InducedGrid}.}%
\label{fig-AnaFunc1Triangular}%
\end{subfigure}%
\hfill%
\begin{subfigure}[t]{\anafunconeimgwidth}%
\includegraphics[width=\linewidth]{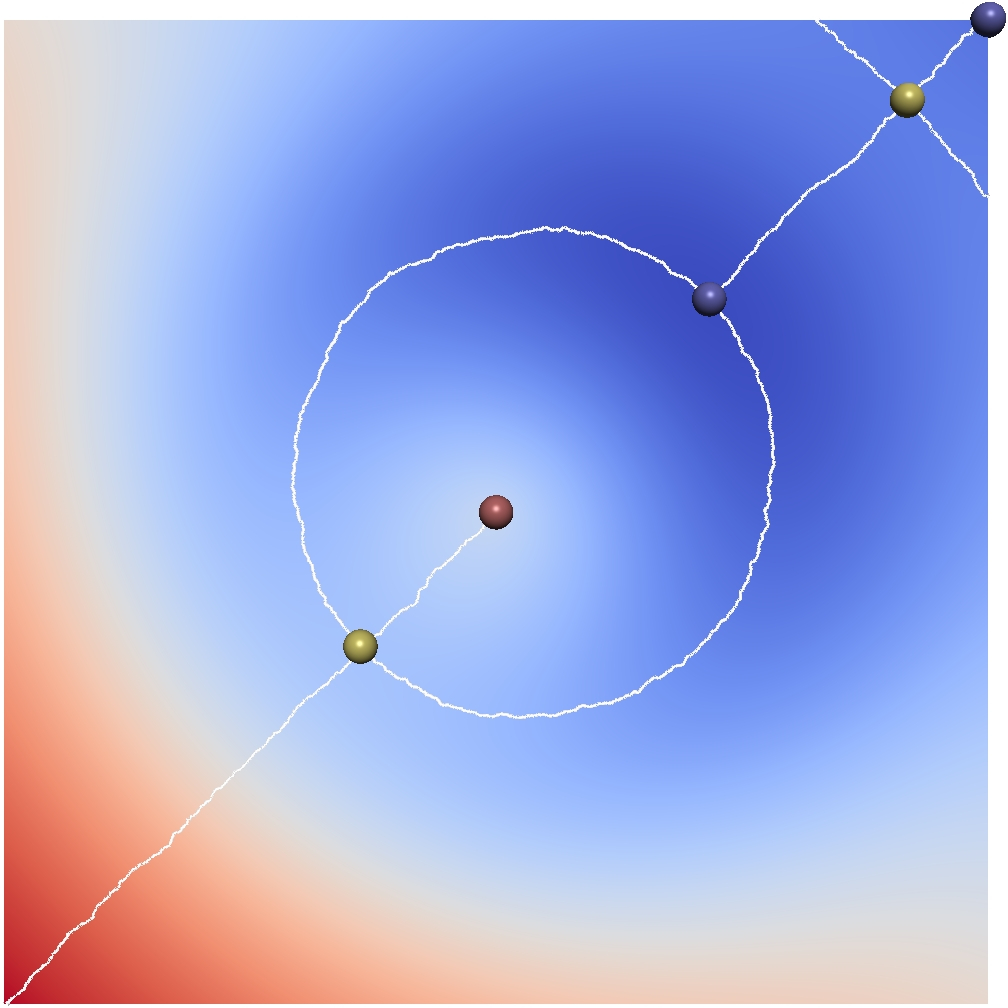}%
\caption{Steepest descent method \cite{robins10} on a triangle grid with random vertex positions and Delaunay triangulation as shown in \Cref{fig-RandomGrid}.}%
\label{fig-AnaFunc1RandomDelaunay}%
\end{subfigure}%
\hfill%
\begin{subfigure}[t]{\anafunconeimgwidth}%
\includegraphics[width=\linewidth]{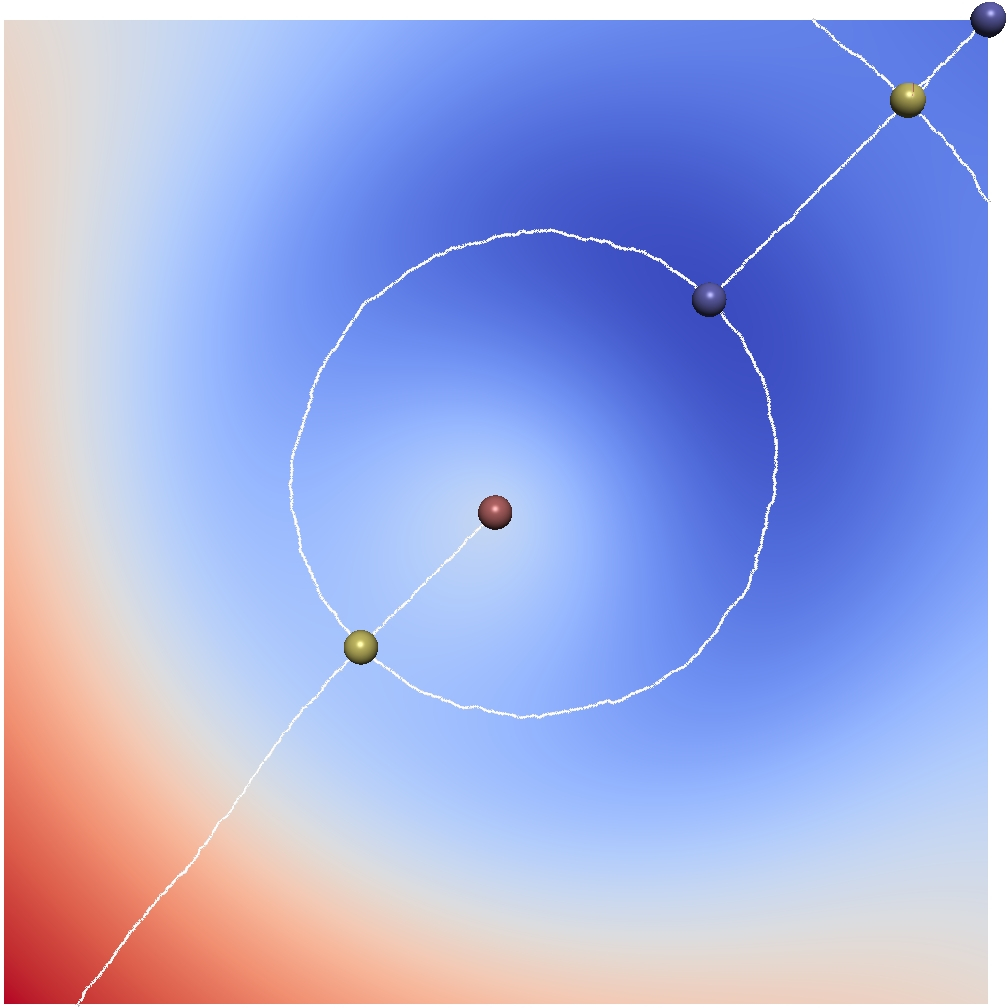}%
\caption{Steepest descent method \cite{robins10} on a triangle grid derived from the uniform grid following our suggestion, as illustrated in \Cref{fig-ModifiedGrid} and discussed in \Cref{section-Suggestion}.}%
\label{fig-AnaFunc1MSSuggestion}%
\end{subfigure}%
\caption{The function from \Cref{formula-AnaFunc1} has been sampled on different types of grids and the Morse-Smale complex has been extracted using different methods. The goal of \emph{accurate geometry} is to achieve a geometric embedding that is similar to the one from continuous topology. While it is well-known that the steepest descent method is not able to achieve this on uniform grids, we show that it can achieve \emph{accurate geometry} on certain grid types and provide a suggestion of how to convert uniform grids accordingly.}%
\label{fig-AnaFunc1MS}%
\end{figure*}

The term \emph{accurate geometry}
was coined by
Gyulassy et al. \cite{gyulassy12ms}
to refer to discrepancies
between the geometric embeddings
of the continuous and the discrete Morse-Smale complex.
The underlying idea is
that the geometric embedding
of the discrete Morse-Smale complex
should converge to its continuous counterpart,
if the discretization of the domain
becomes infinitely fine.
This is not always the case \cite{gyulassy12ms, reininghaus2012combinatorial},
but quite desirable in a number of applications,
e.g., where separatrices are the main features of interest.
We will discuss the previous work on this topic in the following.

Given a discrete gradient field,
the Morse-Smale complex
can be constructed straightforwardly
using an alternating breadth-first search \cite{robins10}.
The main challenge is rather
to compute the discrete gradient field itself
from image or volumetric data
such that the discrete Morse-Smale complex
will align with its continuous counterpart.
Several methods have been proposed on this subject \cite{lewiner03, gyulassy08, robins10}. 
The first provably correct method
proposed by Robins et al. \cite{robins10}
solves the problem efficiently
and is regarded as the \emph{de facto} standard
for this task:
given an image or volume data that can be thought of as a complex, 
this algorithm grows from the vertices to every cell of the complex by partitioning the whole domain into disjoint lower stars of the vertex.
For each lower star, 
the first gradient vector
is chosen from the possible vectors
based on the \emph{steepest descent} direction.
The rest of the vectors
in this lower star
will be constructed via a procedure called \emph{simple homotopy expansion}
by collecting cells with exactly one uncovered face and choosing the first possible pair 
according to an order.
This procedure will be performed until such cells can no longer be found.
If the expansion is not possible, a cell will be recorded 
as critical and the expansion continues to proceed from this cell until every cell in the lower star either belongs to a vector or is critical.

However,
Gyulassy et al. \cite{gyulassy12ms}
and
Reininghaus et al. \cite{reininghaus2012combinatorial}
show independently
that the choice of the \emph{steepest descent} direction 
leads to a non-convergent geometric embedding of the separatrices,
i.e., inaccurate geometry.
This is due to the local and greedy nature of the approach.
Consider the following analytic function 
$f: [0,2]^2 \rightarrow \RR $ from Reininghaus et al. \cite{reininghaus2012combinatorial}
\begin{equation}
f(x,y) = \exp(-2(\sqrt{x^2 + y^2}-1)^2) - 0.3(x+y).
\label{formula-AnaFunc1}
\end{equation}
It represents a circle engraved in a tilted plane.
The continuous Morse-Smale complex reveals this shape in \Cref{fig-AnaFunc1Cont}.
To extract the discrete counterpart,
we sampled this function
on a $1024\times 1024$ uniform lattice grid
and
applied the method of Robins et al. \cite{robins10}
using the steepest descent approach.
As \Cref{fig-AnaFunc1Robins} shows,
the result does not capture the circular shape,
neither for this nor any other sampling resolution.

Gyulassy et al. \cite{gyulassy12ms}
and
Reininghaus et al. \cite{reininghaus2012combinatorial}
tackle this problem using a probabilistic approach:
instead of always choosing the steepest-descent vector,
any descending vector can be chosen
following a probability
that relates back to the continuous gradient.
This means,
the vectors in the discrete gradient field
are more likely to be aligned
with the continuous tangent curves.
This rather small modification
to the original method of Robins et al. \cite{robins10}
produces Morse-Smale complexes 
which will converge to the continuous version
as the sampling resolution increases.
The result can be seen in \Cref{fig-AnaFunc1Jans}.

Gyulassy et al. \cite{gyulassy12ms} 
introduced a second approach to geometric accuracy
that takes into account a larger region when constructing the gradient vector field.
It is of global nature and comes with a substantial computational effort.
Gyulassy et al. \cite{gyulassy14conforming,gyulassy19ms} propose methods
to create \emph{conforming} Morse-Smale complexes:
they are not just based on the input scalar data,
but additional information can be supplied by the user
to influence the output.
Specifically, a discrete vector will be created
between two cells,
if these cells have equal value with respect to a map $L$.
Using this map, one can encode extra information for the 
computation or even modification of the Morse-Smale complex,
according to the needs of the application.
Since a strict adherence to the original Morse-Smale complex
is not desired when applying these methods,
we will not discuss them further in this paper,
but rather leave it to future work to investigate the similarities and differences to the steepest descent method.

\section{Topological Inconsistencies}
\label{section-TopologicalInconsistencies}

We can use the probabilistic methods \cite{gyulassy12ms,reininghaus2012combinatorial}
from \Cref{section-AccurateGeometryForMSComplexes}
to achieve geometrical accuracy. 
However, these methods come with a significant drawback
that has not been reported before:
the resulting Morse-Smale complex
differs \changed{topologically} from the one obtained using steepest descent.
Most importantly,
the connectivity of the separatrices
differs, meaning,
different critical points are connected with each other.
This holds true for 2D and 3D data sets,
and the topological differences have a high persistence,
i.e., they can be observed even after radical topological simplification.
Besides these topological differences,
the geometric embedding of the saddle points changes as well.

\changed{For simplicity, we limit our discussions 
to simplicial and cubical complexes,
but they can be extended to the more general cell complex as well.}

\subsection{Identical Parts}
\label{section-InconsistenciesTheories}

First, we identify which parts of the Morse-Smale complex
will \emph{not} differ between different methods.
Specifically, we show
that any gradient vector field computed based on the lower star has a constant number of critical cells for 
each dimension and the positions of certain types of critical cells are fixed, independent of the order in which cells are chosen to proceed.

To simplify the discussion, we take advantage of a useful concept introduced by Robins et al. \cite{robins10}.
The \emph{reduced lower star} of 
a vertex $\vv$, denoted by $R(\vv)$, is the intersection of the lower star of $\vv$ with the sphere centered at $\vv$ with a small radius $r$, 
that is $R(\vv) = \text{St}^{-}(\vv) \cap S(\vv, r)$. The lower star can be thought of as projecting the cells in $\text{St}^{-}(\vv) \backslash \vv$ 
onto the sphere centered at $\vv$. With this map, for $d> 0$, a $d$-cell in $\text{St}^{-}(\vv)\backslash \vv $ is mapped bijectively 
to a $(d-1)$-cell in $R(\vv)$. This map also transforms the algorithm by Robins et al. \cite{robins10} to an analogous algorithm processing on the 
reduced lower stars, performing homotopy expansion starting from the steepest edge.

From the definition of the reduced lower star, we can see that it is closed, i.e. every face of a simplex $\sigma$ of $R(\vv)$
is also in $R(\vv)$. Thus, the number of $d$-cells of $R(\vv)$ is constant. It was shown by Robins et al. \cite{robins10} that every critical 
$(d-1)$-cells processed from $R(\vv)$ is a critical $d$-cell in the original complex, except for the initial $0$-cell of $R(\vv)$. 
It follows that the number of critical $d$-cells produced by the lower star-based methods is constant for $d>0$, independent of the order 
in the expansion step. Furthermore, for any vertices of a complex $K$, the number of elements in its lower star is constant,
and a vertex $\vv$ is marked as a minimum if $\text{St}^{-}(\vv) = \varnothing$. This implies that the number and the position of the 
minima are fixed.

Moving to higher dimensional cells, we break the results into smaller pieces. 
\changed{In the following,
we say that a simplex is \emph{covered}
if it belongs to a discrete vector or is marked 
as critical during the homotopy expansion process,
and \emph{uncovered} otherwise.}
We stay in the language of reduced lower stars and show that 
using homotopy expansion in a $1$-skeleton of this structure resulting in a same set of covered $0$- and $1$-simplices
through the following lemmas. 
\begin{lemma}
  The homotopy expansion of the $1$-skeleton of each $R(\vv)$ covers all of its $0$-simplices.
  \label{lemma-CoverAllPoints}
\end{lemma}
\begin{proof}
  Assuming that there exists one uncovered $0$-simplex $\alpha$ during the expansion process. Then this uncovered simplex connects
  to a covered $0$-simplex via a $1$-simplex $\beta$. The homotopy expansion will pair $\{\alpha, \beta \}$ during the process.
\end{proof}

\begin{lemma}
  Given $R(\vv)$, the set of all $1$-simplices covered by the homotopy expansion in $R(\vv)$ is the same, independent of the starting 
  point of the process.
  \label{lemma-CoverSameEdges}
\end{lemma}
\begin{proof}
  A $1$-simplex would not be paired during the homotopy expansion if its two $0$-simplices have been paired with other $1$-simplices.
  This behavior can only happen if that $1$-simplex is a part of a cycle. Two intersecting cycles would prevent the pairing of their
  last $1$-simplices, marking them as critical cells. All other $1$-simplices will be paired during this homotopy expansion as it covers
  all of the $0$-simiplices owing to the result of \Cref{lemma-CoverAllPoints}. It implies that a $1$-simplex is unpaired if and only 
  if it is maximal in a cycle. This defines the fixed set of covered $1$-simplex by homotopy expansion of the $1$-skeleton of $R(\vv)$.
\end{proof}

The result of \Cref{lemma-CoverSameEdges} concerns only up to the coverage of the $1$-simplices. 
The following lemma ensures that after performing the homotopy expansion on the $1$-skeleton of a complex, the result gradient 
field of that complex is fixed. 
\begin{lemma}
  For each $R(\vv)$, the homotopy expansion of its $1$-skeleton can be done before expanding into $2$-simplices without changing 
  the result gradient field.
  \label{lemma-FixedGradField}
\end{lemma}
\begin{proof}
  Supposing that we completed expanding in the $1$-skeleton of $R(\vv)$. Assigning the first vector $\{\alpha^{(1)}, \beta^{(2)}\}$ 
  requires that $\alpha^{(1)}$ is the only face of $\beta^{(2)}$ that is unpaired during the expansion of the $1$-skeleton. 
  Thus, no vector $\{\gamma^{(0)}, \alpha^{(1)}\}$ is chosen during the first expansion since this would imply the expansion 
  had not been finished. 
\end{proof}

From \Cref{lemma-FixedGradField}, mapping to the original lower star, we get the following result.
\begin{theorem}
  For any vertex $\vv$ of a complex $K$, the gradient vector fields between the $d$- and $(d+1)$-cells in $\text{St}^{-}(\vv)$
  are fixed for $d \geq 2$, independent of the order of processing the cells.
  \label{theorem-FixedGradField}
\end{theorem}

In the arguments for \Cref{lemma-CoverSameEdges}, we can see that the critical $1$-simplices which are unpaired during the 
homotopy expansion of the $1$-skeleton will also not be paired during the expansion of $2$-skeleton of $R(\vv)$, and their positions 
will also be fixed. From \Cref{theorem-FixedGradField}, we can conclude the same for higher dimensional cells. The following results
summarize our observations.
\begin{corollary}
  For $d \geq 2$, the position of critical $d$-cells in $\text{St}^{-}(\vv)$, for every vertex $\vv$, are independent of the chosen order
  in the computing process of the gradient vector field.
  \label{corollary-FixedCritical}
\end{corollary}

\begin{corollary}
  For $d \geq 2$, each separatrix connecting a pair of $d$- and $(d+1)$-critical cells retains its connectivity, independent
  of the chosen order in the computing process of the gradient vector field. 
  \label{corollary-FixedSeparatrices}
\end{corollary}

We conclude that the positions of the minima and maxima in 2D and 3D data sets
are identical between the different methods.
The same holds for the $2$-saddle points in 3D data sets.
We further conclude
that the connectivity of the separatrices
between 2-saddles and maxima in 3D data sets
is identical between the different methods.

\subsection{Different Parts}
\label{section-InconsistenciesExperiments}

The saddle points ($1$-saddle points in the case of three dimensions)
did not appear in the results of
\Cref{section-InconsistenciesTheories}.
Indeed, with different ways of choosing the first vector in the computation (steepest descent or probabilistically), 
they can have different positions. We demonstrate this behavior by applying the steepest descent
method and the probabilistic method on a scalar field given by a small randomly generated $4\times 4$ matrix
\begin{align}
  A = \begin{pmatrix}
    9 & 8 & 7 & 13\\
    1 & 6 & 10 & 0\\
    12 & 14 & 2 & 4\\
    3 & 11 & 5 & 15
  \end{pmatrix}.
  \label{formula-Matrix1}
\end{align}
The differences in the positions of the saddle points are shown in 
\Cref{fig-Mat1Critical}. We can see from the two figures that the positions of the minima and the maximum are fixed whereas 
\changed{three saddle points} 
changed \changed{their positions}.
This is in agreement with the theoretical results stated earlier. It should 
be noted that while the saddle points can move in different iterations of the probabilistic approach, their movements are restricted to their
respective lower stars.

\begin{figure}[t]%
\begin{subfigure}{0.48\linewidth}%
\includegraphics[width=\textwidth]{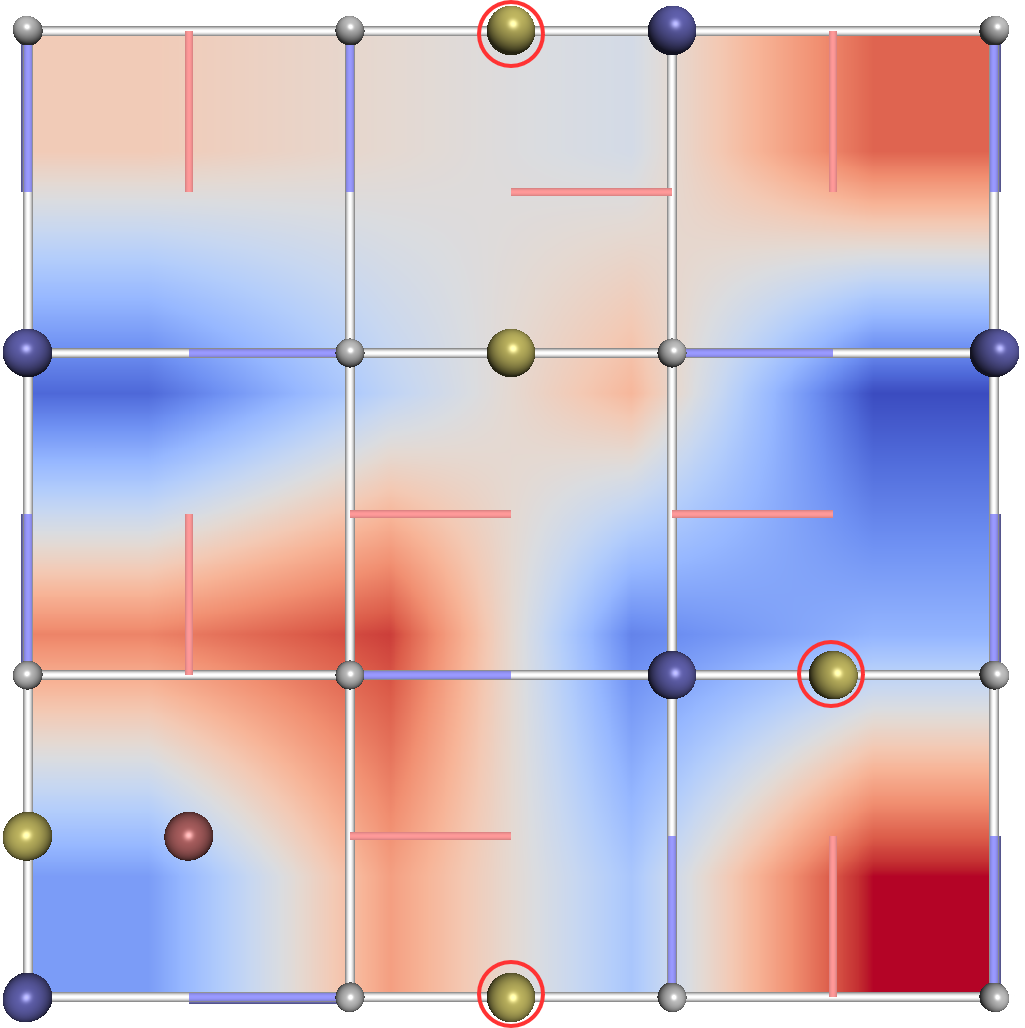}%
\caption{Steepest descent \cite{robins10}}%
\label{fig-Mat1RobinsCritical}%
\end{subfigure}%
\hfill%
\begin{subfigure}{0.48\linewidth}%
\includegraphics[width=\textwidth]{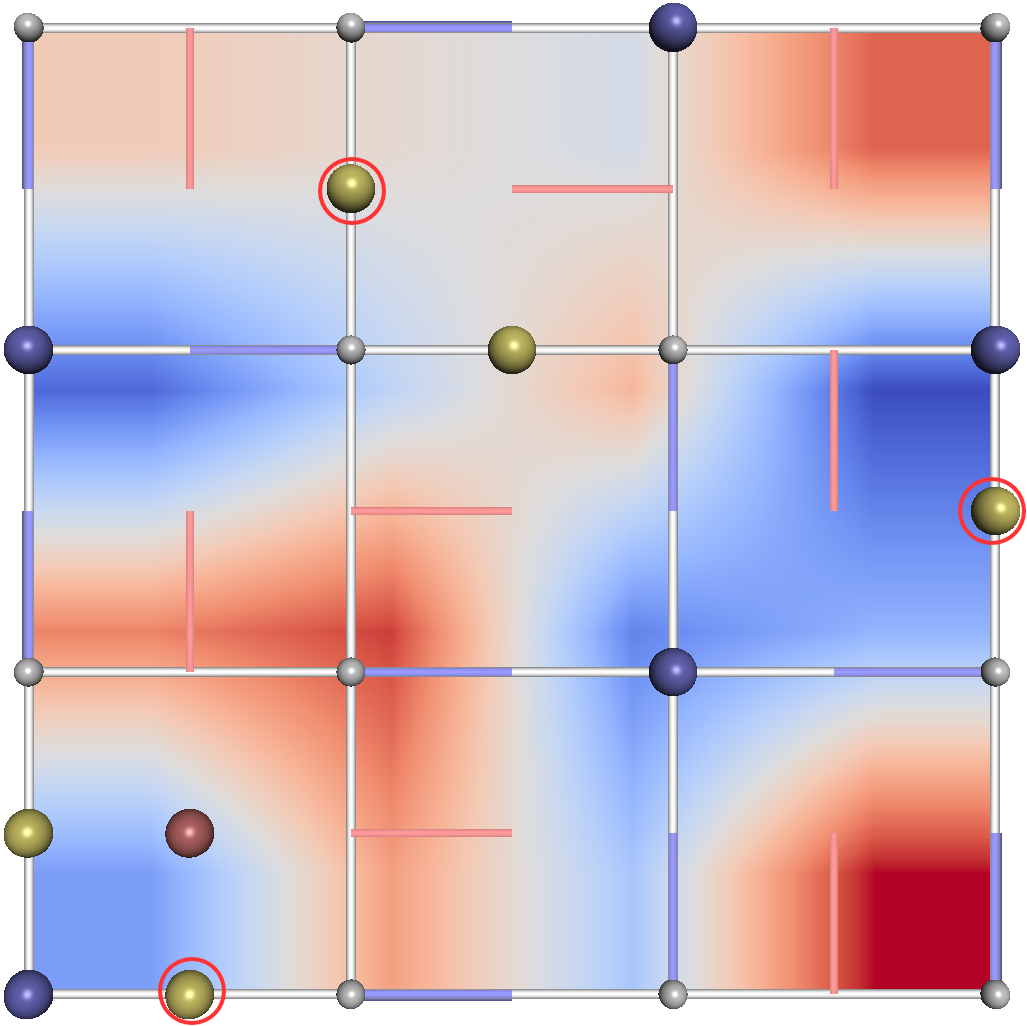}%
\caption{Probabilistically chosen \cite{gyulassy12ms, reininghaus2012combinatorial}}%
\label{fig-Mat1JansCritical}%
\end{subfigure}%
\caption{Given the same input data from \Cref{formula-Matrix1}, 
the two different methods for computing the discrete gradient field 
can result in different positions for saddle points as indicated by the red circles.
\changed{The white spheres, white lines, and squares represent
the $0$-cells, $1$-cells, and $2$-cells.
The blue and red segments illustrate 
the pairings between 
the $0$- and $1$-cells
as well as the
$1$- and $2$-cells, respectively.}}%
\label{fig-Mat1Critical}%
\end{figure}

\begin{figure}%
\begin{subfigure}{0.48\linewidth}%
\includegraphics[width=\textwidth]{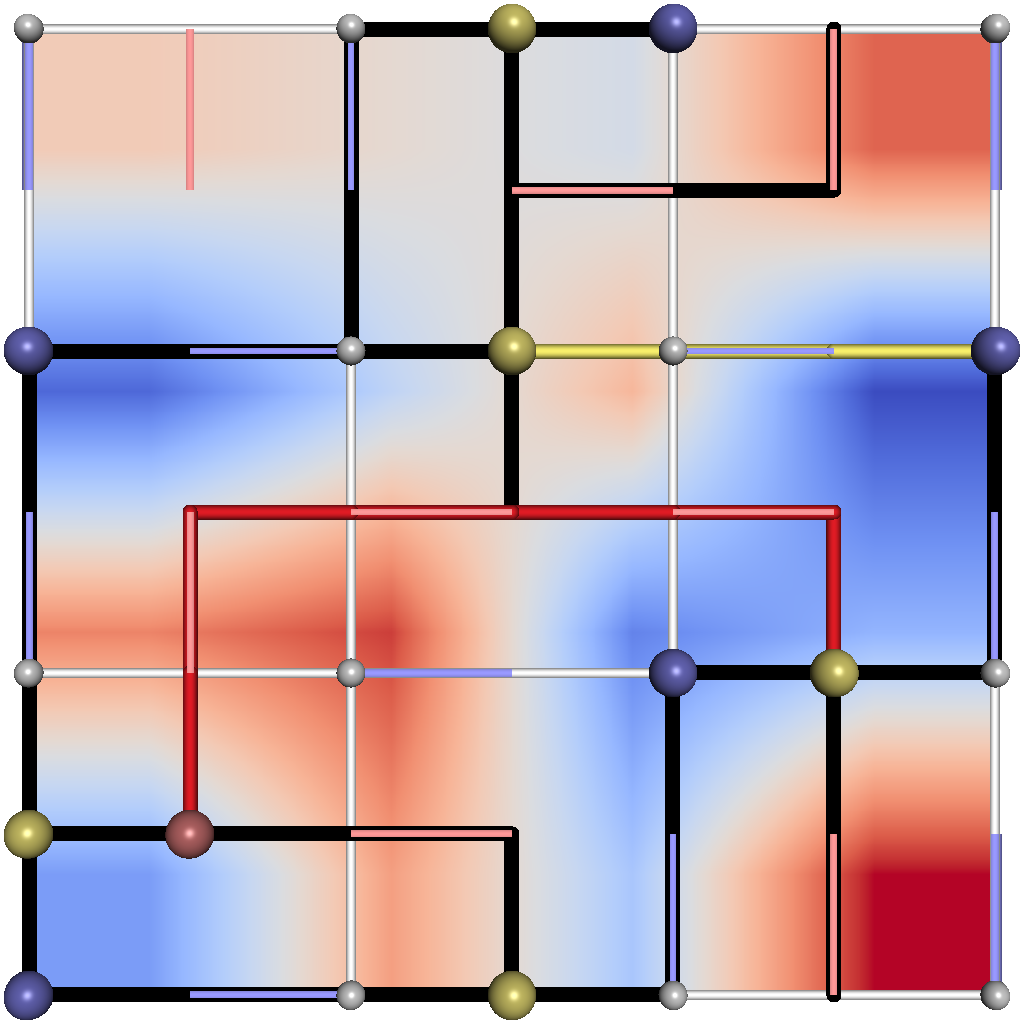}%
\caption{Steepest descent \cite{robins10}}%
\label{fig-Mat1RobinsSeparatrices}%
\end{subfigure}%
\hfill%
\begin{subfigure}{0.48\linewidth}%
\includegraphics[width=\textwidth]{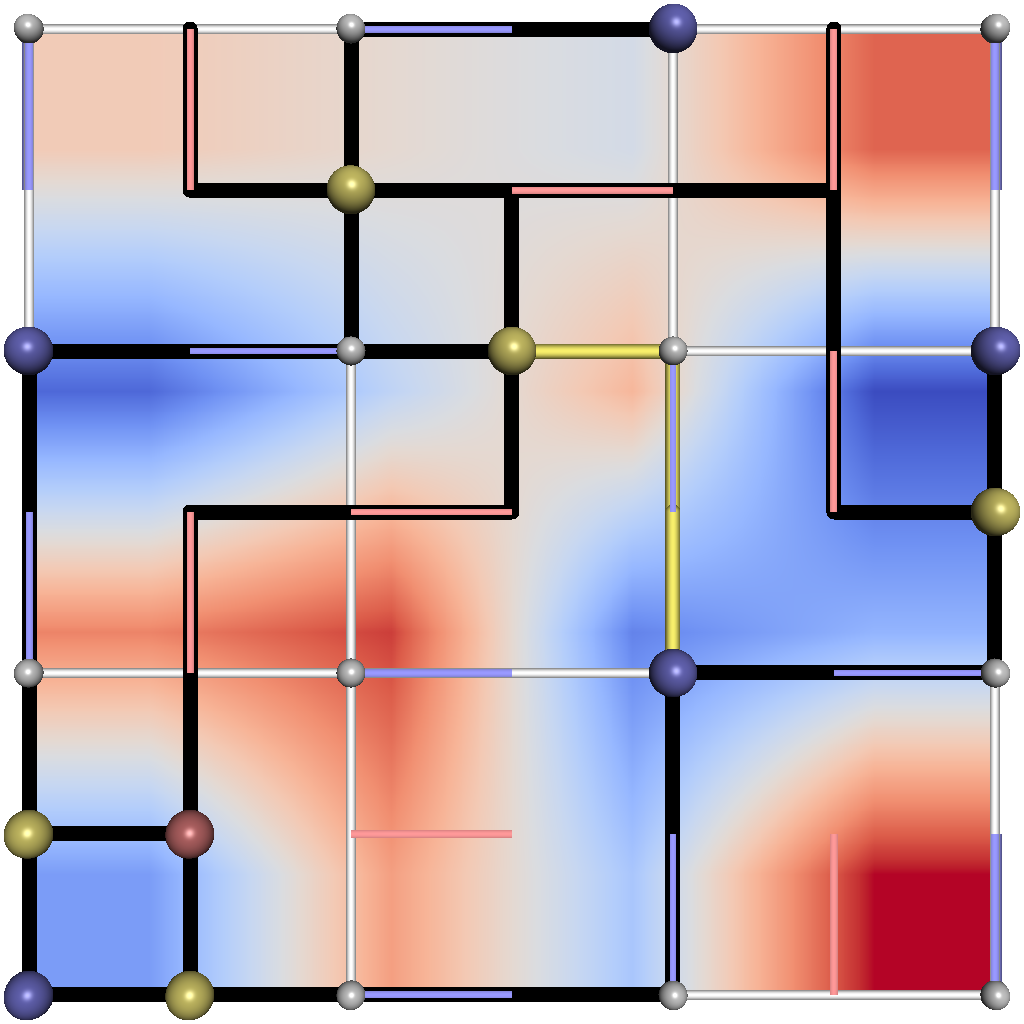}%
\caption{Probabilistically chosen \cite{gyulassy12ms, reininghaus2012combinatorial}}%
\label{fig-Mat1JansSeparatrices}%
\end{subfigure}%
\caption{Given the same input data from \Cref{formula-Matrix1}, 
the two different methods for computing the discrete gradient field 
can result in a different connectivity of the separatrices. 
The separatrices connecting to the same critical points are colored black. 
The separatrices connecting to different minima are highlighted in yellow. 
The red separatrix only presents for the steepest descent method, and is missing from the output of the probabilistic method.}%
\label{fig-Mat1Separatrices}%
\end{figure}

The changes in the position of the saddle points
further cause separatrices
to connect to different critical points
as illustrated in \Cref{fig-Mat1Separatrices}.
We can easily see
the top-most saddle
connects to a different minimum,
and that another saddle looses its connection to a maximum.
This observation
can be explained
by the difference
in the gradient vector fields.
Notably,
the result of \Cref{corollary-FixedSeparatrices}
only holds for separatrices in higher layers,
whereas this example shows separatrices
between 0-cells, 1-cells, and 2-cells,
for which \Cref{corollary-FixedSeparatrices}
does not hold.

As the changes already happen in two dimensions, it is straightforward that these changes can also appear in three dimensions. Indeed, 
consider the scalar field given by the small randomly generated tensor
\begin{align}
  B = \begin{pmatrix}
    \begin{pmatrix}
      14 & 18 & 9\\
      6 & 24 & 15\\
      23 & 4 & 7
    \end{pmatrix},
    \begin{pmatrix}
      10 & 20 & 0\\
      21 & 1 & 11\\
      2 & 26 & 25
    \end{pmatrix},
    \begin{pmatrix}
      12 & 13 & 8\\
      19 & 22 & 3\\
      17 & 5 & 16
    \end{pmatrix}
  \end{pmatrix}.
  \label{formula-Tensor 1}
\end{align}
The inconsistencies of the positions of the $1$-saddles (green) together with the connectivity of the separatrices 
are illustrated in \Cref{fig-Tensor1Critical}.

\begin{figure}
  \begin{subfigure}{0.48\linewidth}
      \includegraphics[width=\textwidth]{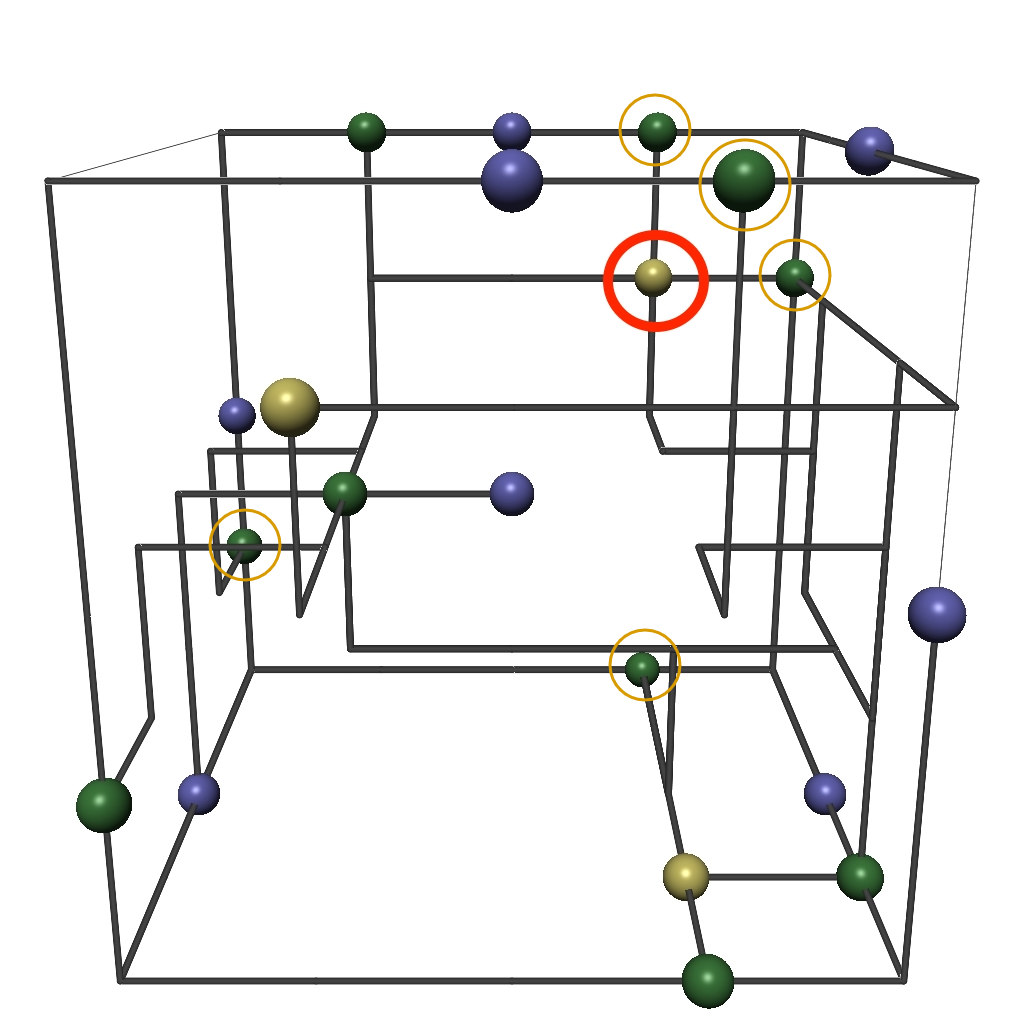}
      \caption{Steepest descent \cite{robins10}}
      \label{fig-Tensor1RobinsSeparatrices}
  \end{subfigure}
  \hfill
  \begin{subfigure}{0.48\linewidth}
      \includegraphics[width=\textwidth]{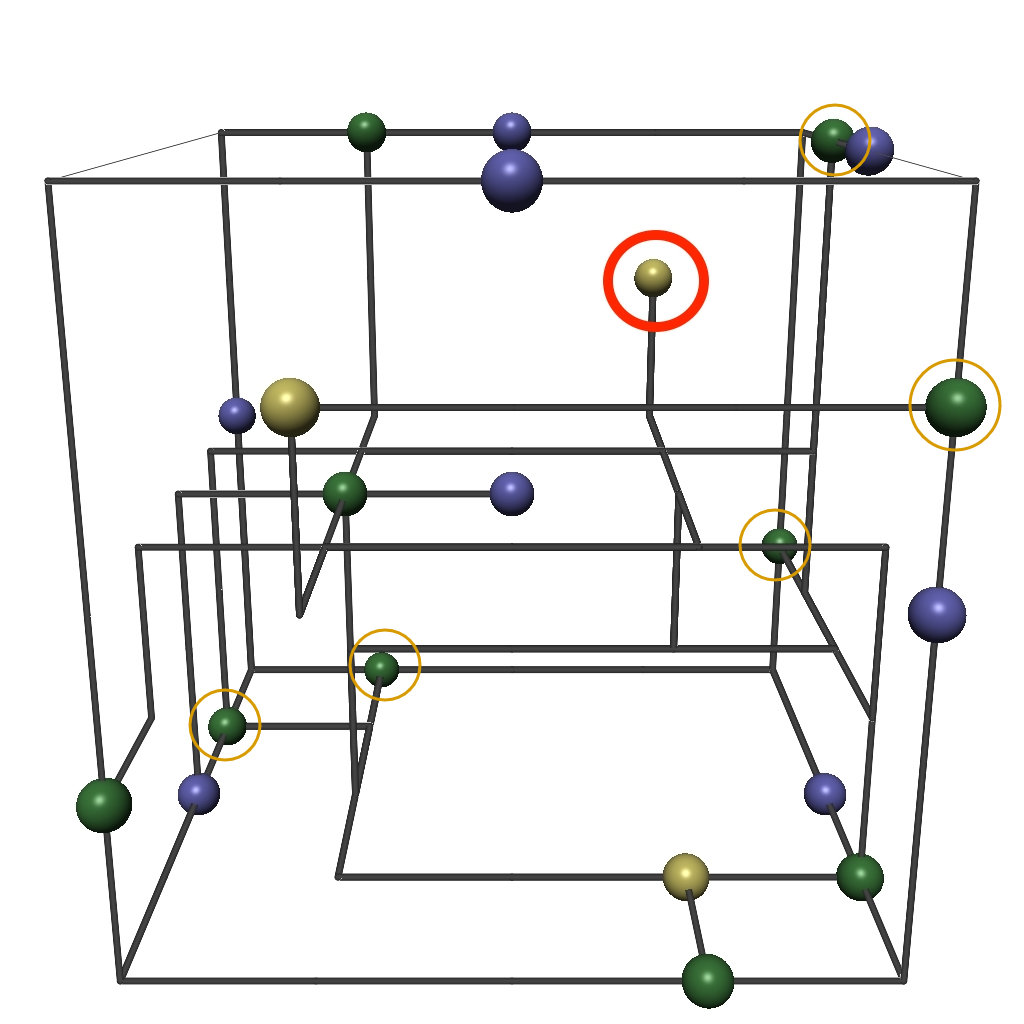}
      \caption{Probabilistically chosen \cite{gyulassy12ms, reininghaus2012combinatorial}}
      \label{fig-Tensor1JansSeparatrices}
  \end{subfigure}
  \caption{Morse-Smale complexes of the simple 3D data set from \Cref{formula-Tensor 1} computed using different methods. Many $1$-saddle points (green) have different positions (marked by the orange circles). This also leads to a different connectivity of the separatrices (marked by the red circle).}
  \label{fig-Tensor1Critical}
\end{figure}

Next, we show that 
these issues also appear in larger and smooth data sets.
We constructed a data set of
size $40 \times 40$
from \Cref{formula-Matrix1}
using bilinear interpolation.
We can see in \Cref{fig-Mat1Interpolate}
that the connectivity between
critical points is different
as highlighted in the figure. 
This shows that topological changes appear in smooth data sets as well.

These differences do not only happen in a few selected scalar fields.
We randomly generated $1000$ 2D data sets of size $4 \times 4$, $8 \times 8$, and $16 \times 16$.
Then we counted for how many of those data sets
we can observe a positional change of at least one saddle,
and a connectivity change of at least one separatrix.
The results are summarized in \Cref{table-SaddlePosition2d}.
We can see that 
these changes appear in
almost all data sets of size $8 \times 8$ and larger,
and in many smaller data sets as well,
\changed{even after interpolating the $4 \times 4$ data sets to higher resolutions.}
Automatically detecting a connectivity change of a separatrix
can be done as follows:
for each separatrix in a Morse-Smale complex,
identify the two lower stars where it originates/ends.
Then try to find a separatrix in the other Morse-Smale complex
which originates/ends in the same pair of lower stars.
If it cannot be found, we record a connectivity change.
This procedure is rooted in the fact that the positional change of a saddle point is restricted to the respective lower star and hence its separatrices need to originate/end there.

\begin{table}%
\centering%
\resizebox{\linewidth}{!}{%
\begin{tabular}{lcc}
\toprule
\textbf{Data set} & \makecell{\textbf{\# cases with}\\\textbf{positional change}} 
& \makecell{\textbf{\# cases with}\\\textbf{connectivity change}}\\
\midrule
Random $4 \times 4$ & $603$ & $235$\\
Random $8 \times 8$ & $996$ & $942$\\
Random $16 \times 16$ & $1000$ & $1000$\\
Interpolated $8 \times 8$ & $567$ & $533$\\
Interpolated $16 \times 16$& $581$ & $555$\\
\bottomrule
\end{tabular}}%
\caption{Number of cases with differences in the positions of the saddle points 
and in the connectivity of the separatrices between the two methods when applied to $1000$ randomly generated 2D data sets.}
\label{table-SaddlePosition2d}
\end{table}

\begin{figure}[!t]%
\begin{minipage}{\linewidth}%
\begin{subfigure}[t]{0.48\linewidth}%
\includegraphics[width=\textwidth]{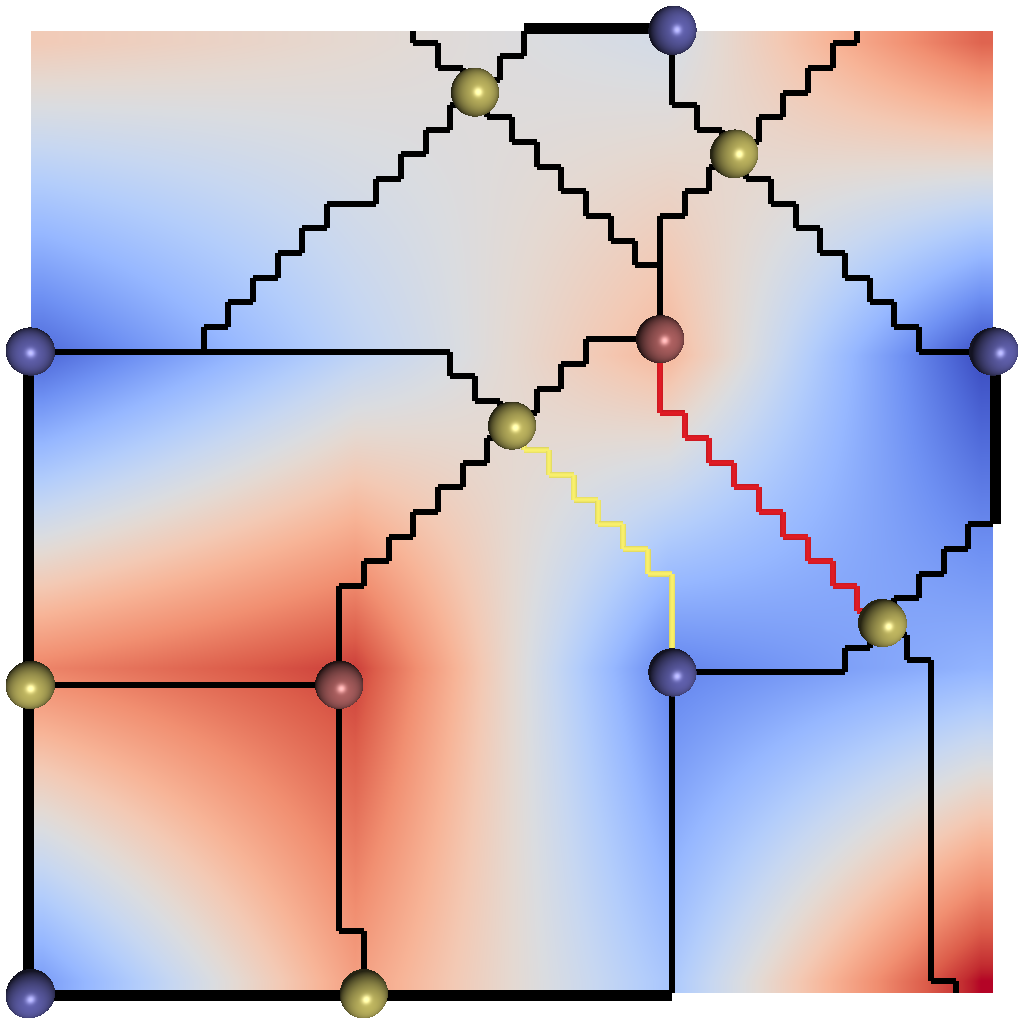}%
\caption{Steepest descent \cite{robins10}}%
\label{fig-Mat1InterpolateRobins}%
\end{subfigure}%
\hfill%
\begin{subfigure}[t]{0.48\linewidth}%
\includegraphics[width=\textwidth]{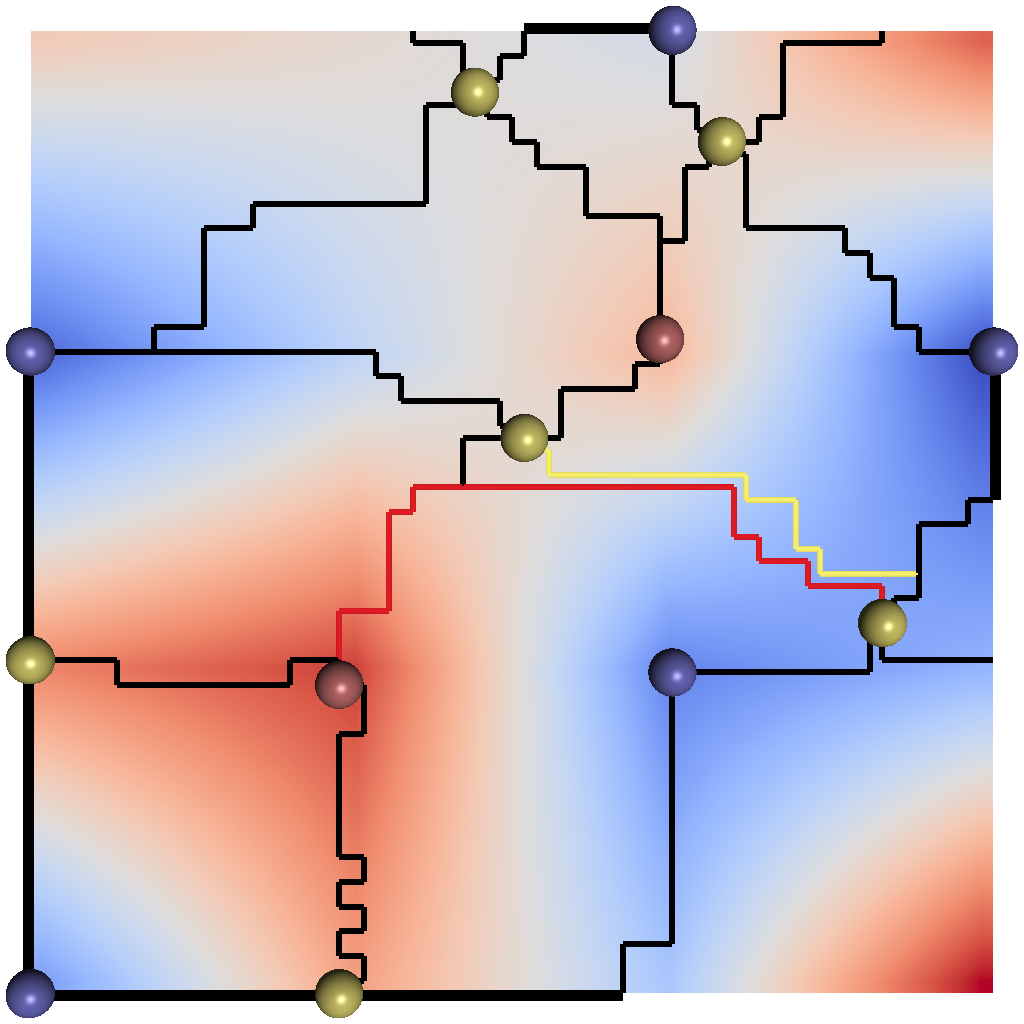}%
\caption{Probabilistically chosen \cite{gyulassy12ms, reininghaus2012combinatorial}}%
\label{fig-Mat1InterpolateJans}%
\end{subfigure}%
\caption{The data set from \Cref{formula-Matrix1} has been bilinearly interpolated on a $40\times 40$ uniform grid. The Morse-Smale complexes have been extracted using the different methods and we observe connectivity changes for some separatrices (highlighted in red and yellow).}
\label{fig-Mat1Interpolate}
\end{minipage}%
\\[2\baselineskip]
\begin{minipage}{\linewidth}%
\begin{subfigure}[t]{0.48\linewidth}%
\includegraphics[width=\textwidth]{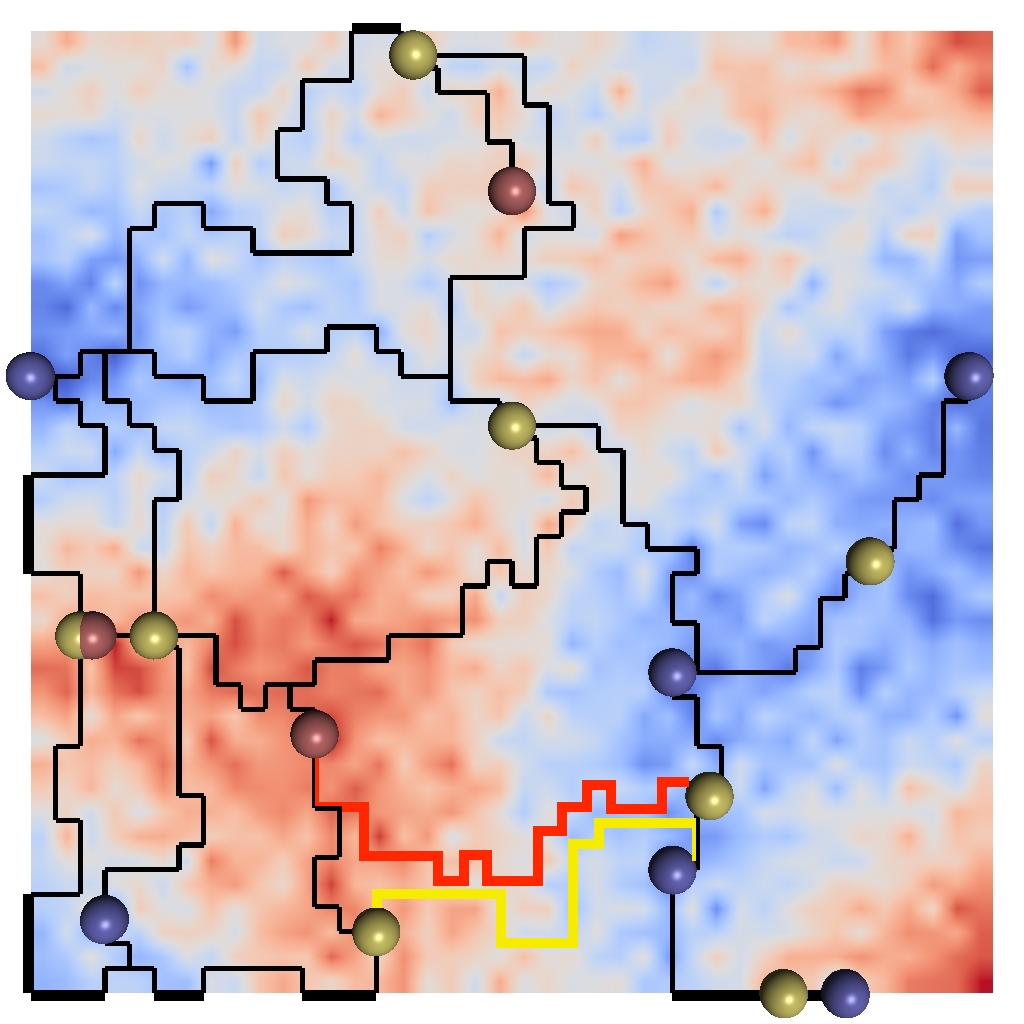}
\caption{Steepest descent \cite{robins10}}%
\label{fig-Mat1NoiseSimplified}%
\end{subfigure}%
\hfill%
\begin{subfigure}[t]{0.48\linewidth}%
\includegraphics[width=\textwidth]{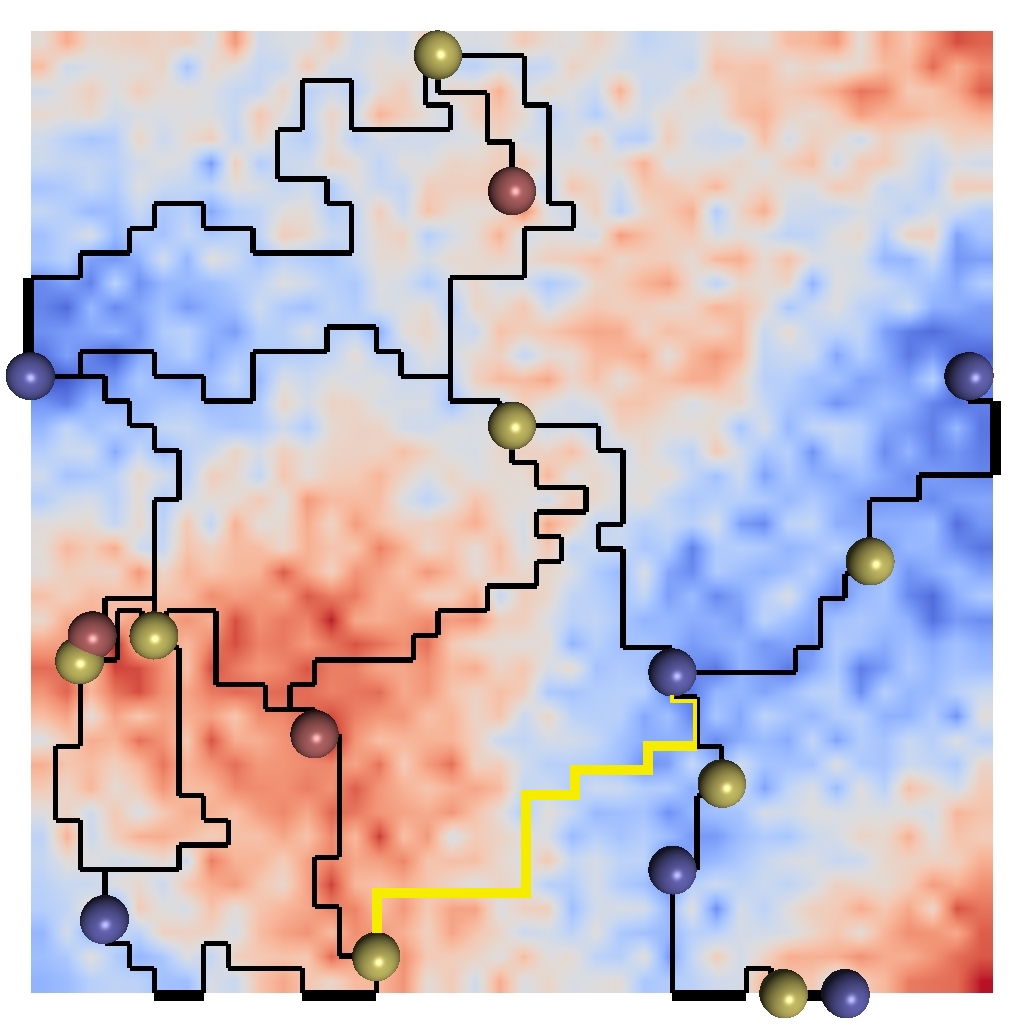}%
\caption{Probabilistically chosen \cite{gyulassy12ms, reininghaus2012combinatorial}}%
\label{fig-Mat1NoiseJanSimplified}%
\end{subfigure}%
\caption{Noise has been added to the data set from \Cref{fig-Mat1Interpolate}. The Morse-Smale complexes produced by the two methods show differences in separatrix connectivity (red and yellow) even at a very high level of topological simplification corresponding to $45\%$ of the data range.}%
\label{fig-Mat1Noise}%
\end{minipage}%
\\[2\baselineskip]
\begin{minipage}{\linewidth}%
\begin{subfigure}[t]{0.48\linewidth}%
\includegraphics[width=\textwidth]{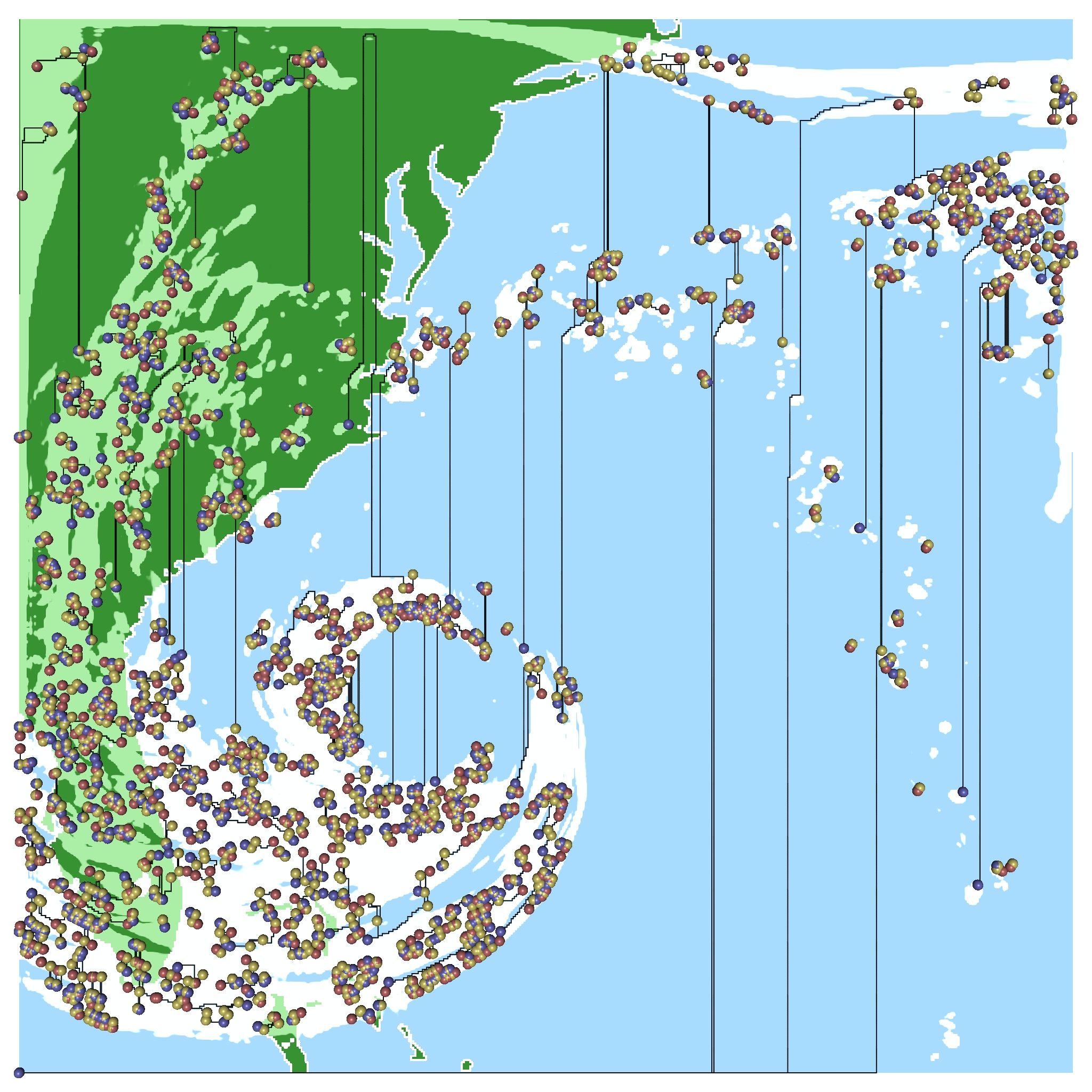}%
\caption{Steepest descent \cite{robins10}}%
\label{fig-CloudRobins}%
\end{subfigure}%
\hfill%
\begin{subfigure}[t]{0.48\linewidth}%
\includegraphics[width=\textwidth]{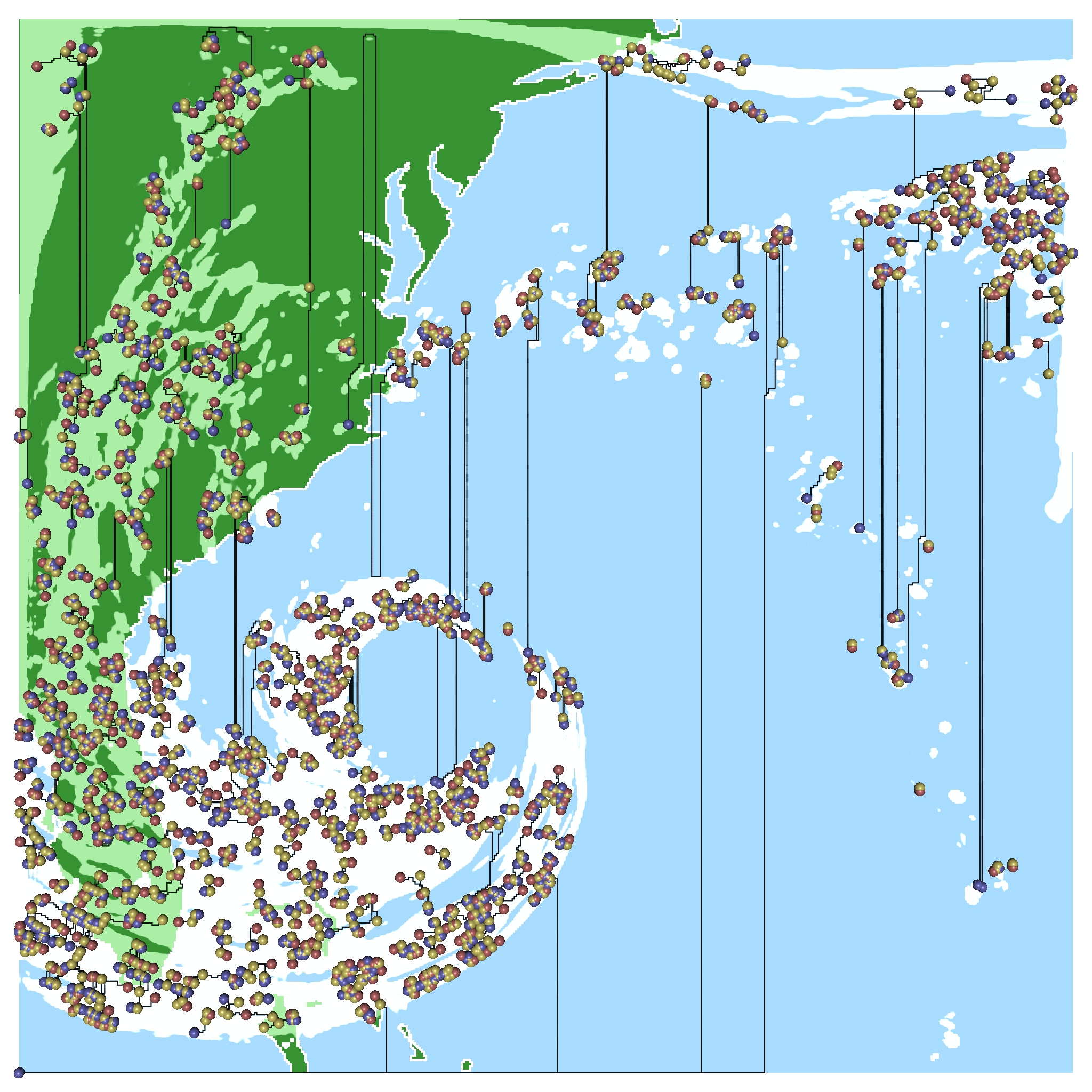}%
\caption{Probabilistically chosen \cite{gyulassy12ms, reininghaus2012combinatorial}}%
\label{fig-CloudJans}%
\end{subfigure}%
\caption{The Morse-Smale complex of the Hurricane Isabel data set has been extracted using the two different methods. Only the separatrices with differences in connectivity are shown together with their endpoints. This is about $10\%$ of all separatrices.}
\label{fig-CloudIsabel}
\end{minipage}%
\end{figure}

We repeat this experiment on a noisy data set,
which we obtained by adding noise to the previous data set.
We find $909$ critical points in this data set.
We apply a topological simplification with a persistence threshold of $45\%$ of the data range.
The simplified Morse-Smale complexes then have $17$ critical points.
We can see in \Cref{fig-Mat1Noise}
that connectivity changes still persist under these conditions -- even at this very high level of persistence.

To show that these differences also exist in real-life data sets,
we performed our experiments on the Hurricane Isabel data set.\footnote{This data set can be downloaded at \url{https://www.earthsystemgrid.org/dataset/isabeldata.html}}
The data set 
represents different atmospheric variables
over different time steps.
We chose the $41$-st time step 
and the \emph{total cloud} variable.
Further,
we extracted the $50$-th slice in $z$-direction
and use it as a 2D data set with a $500 \times 500$ uniform grid.
The differences in the results of the two methods can be seen in \Cref{fig-CloudIsabel}. 
A total of $2087$ out of $20159$ separatrices with different connectivity were found.
Additionally, $10135$ critical points were detected, among them, $983$ out of $5067$ saddle points have different positions.

\paragraph{Summary}
We have shown that the probabilistic methods of
Gyulassy et al. \cite{gyulassy12ms}
and
Reininghaus et al. \cite{reininghaus2012combinatorial}
are successful in producing accurate geometry,
but fail at faithfully reproducing
the topology
of the method by Robins et al. \cite{robins10},
which follows the steepest descent
and is regarded as the \emph{de facto} standard
for computing discrete gradient fields.
These discrepancies arise in many different scenarios
and may affect the conclusions drawn
from topological analyses of data sets.

This leads us in the following
to investigating alternatives to the probabilistic methods
in order to find solutions that provide accurate geometry and consistent topology at the same time.

\section{Effect of Sampling Methods}
\label{section-EffectOfSampling}

\resetlength{\mygridwidth}{\ltwopicwidth}
\begin{figure}[t]%
\begin{subfigure}[t]{\mygridwidth}%
\includegraphics[width=\mygridwidth]{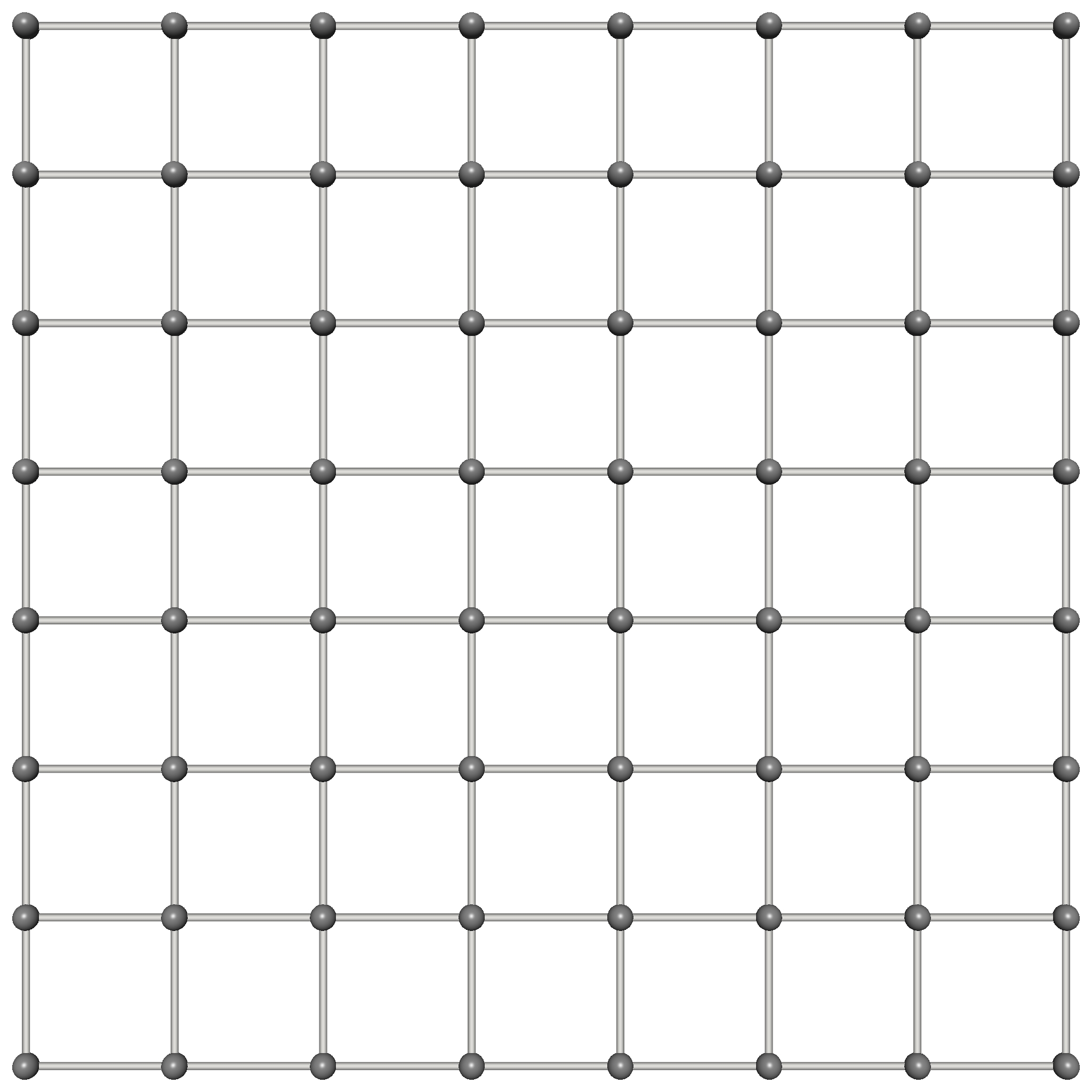}%
\caption{Uniform grid.}%
\label{fig-LatticeGrid}%
\end{subfigure}
\hfill%
\begin{subfigure}[t]{\mygridwidth}%
\includegraphics[width=\mygridwidth]{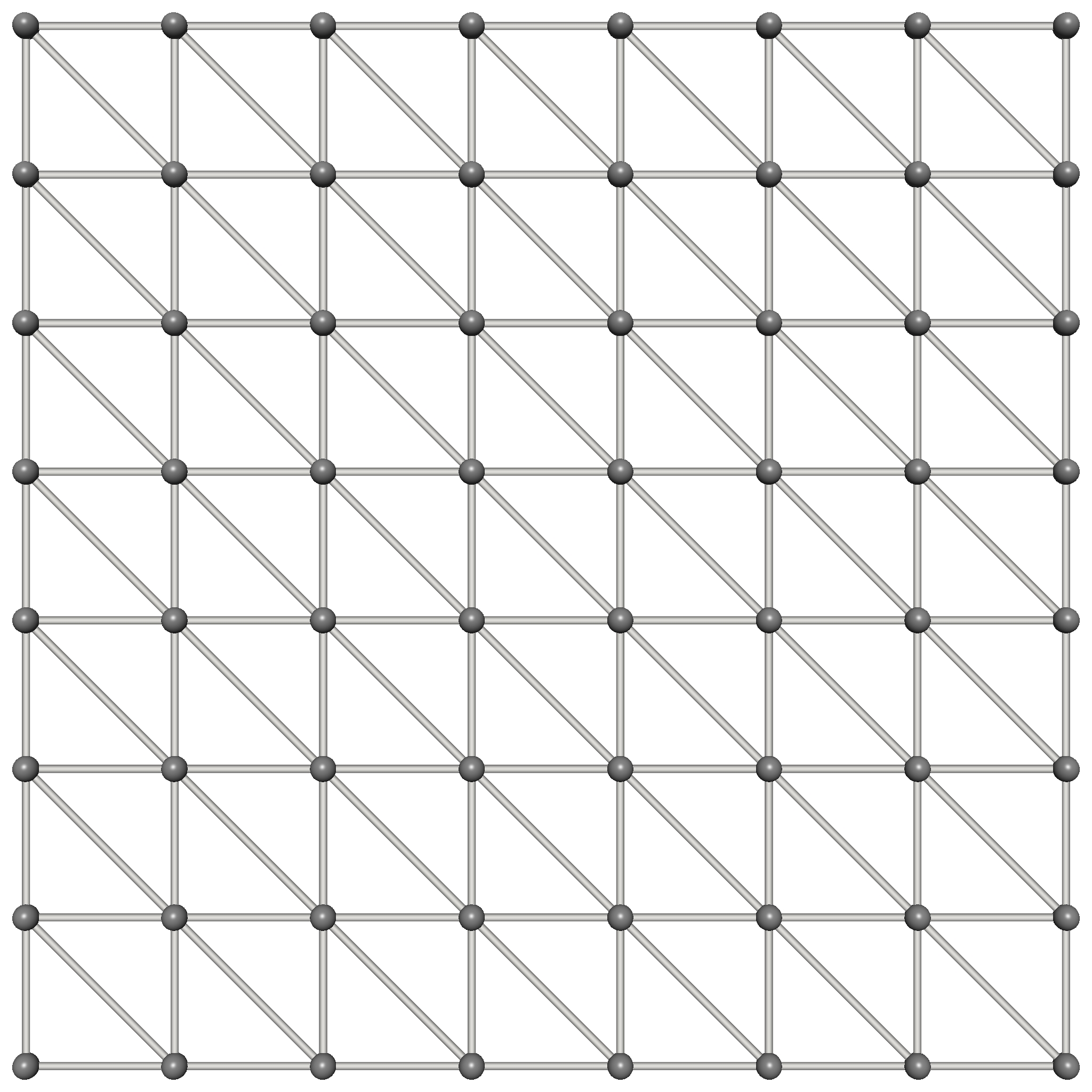}%
\caption{Triangle grid obtained by diagonally dividing each cell of a uniform grid once.}%
\label{fig-InducedGrid}%
\end{subfigure}%
\\%
\begin{subfigure}[t]{\mygridwidth}%
\includegraphics[width=\mygridwidth]{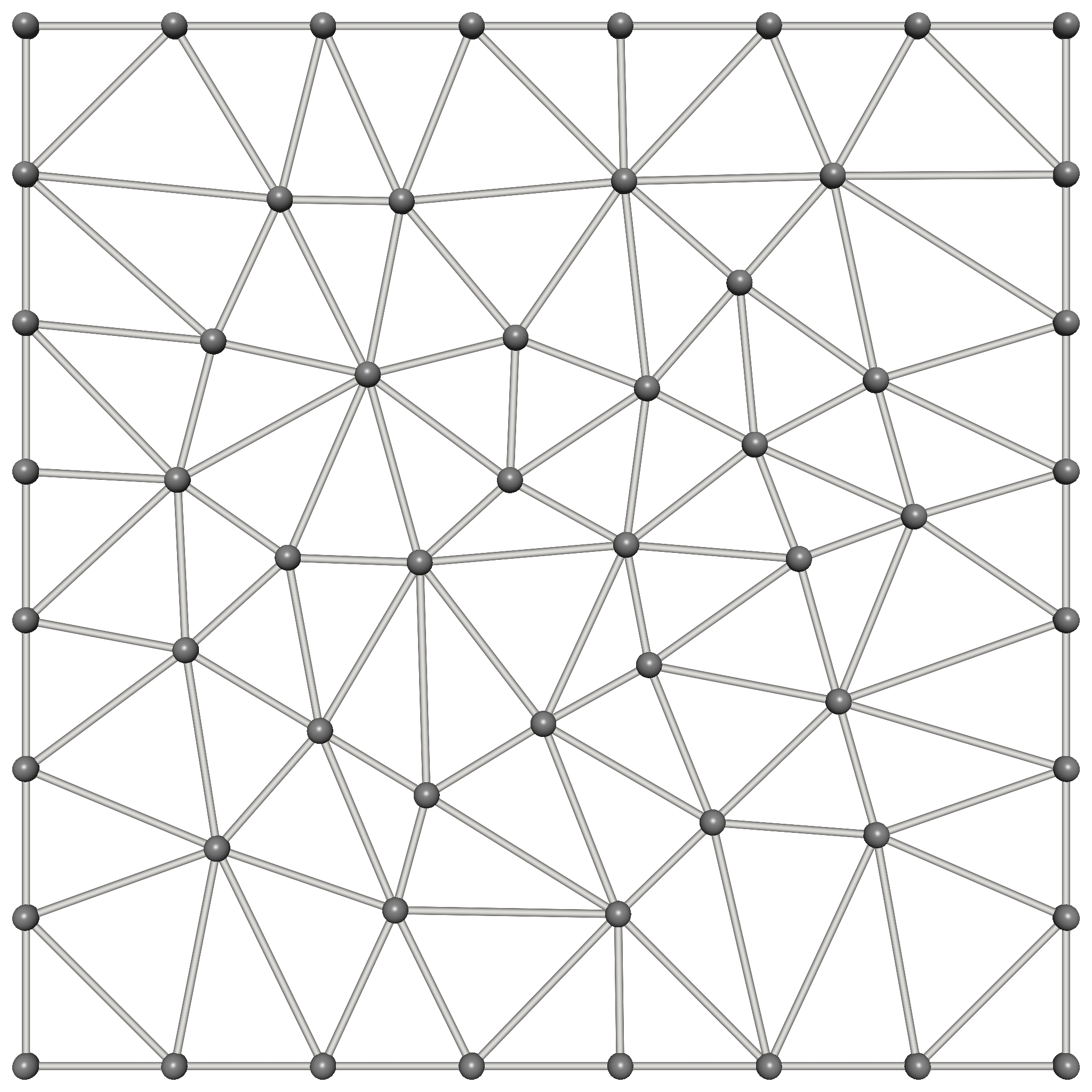}%
\caption{Triangle grid with random vertex positions and Delaunay triangulation.}%
\label{fig-RandomGrid}%
\end{subfigure}%
\hfill%
\begin{subfigure}[t]{\mygridwidth}%
\includegraphics[width=\mygridwidth]{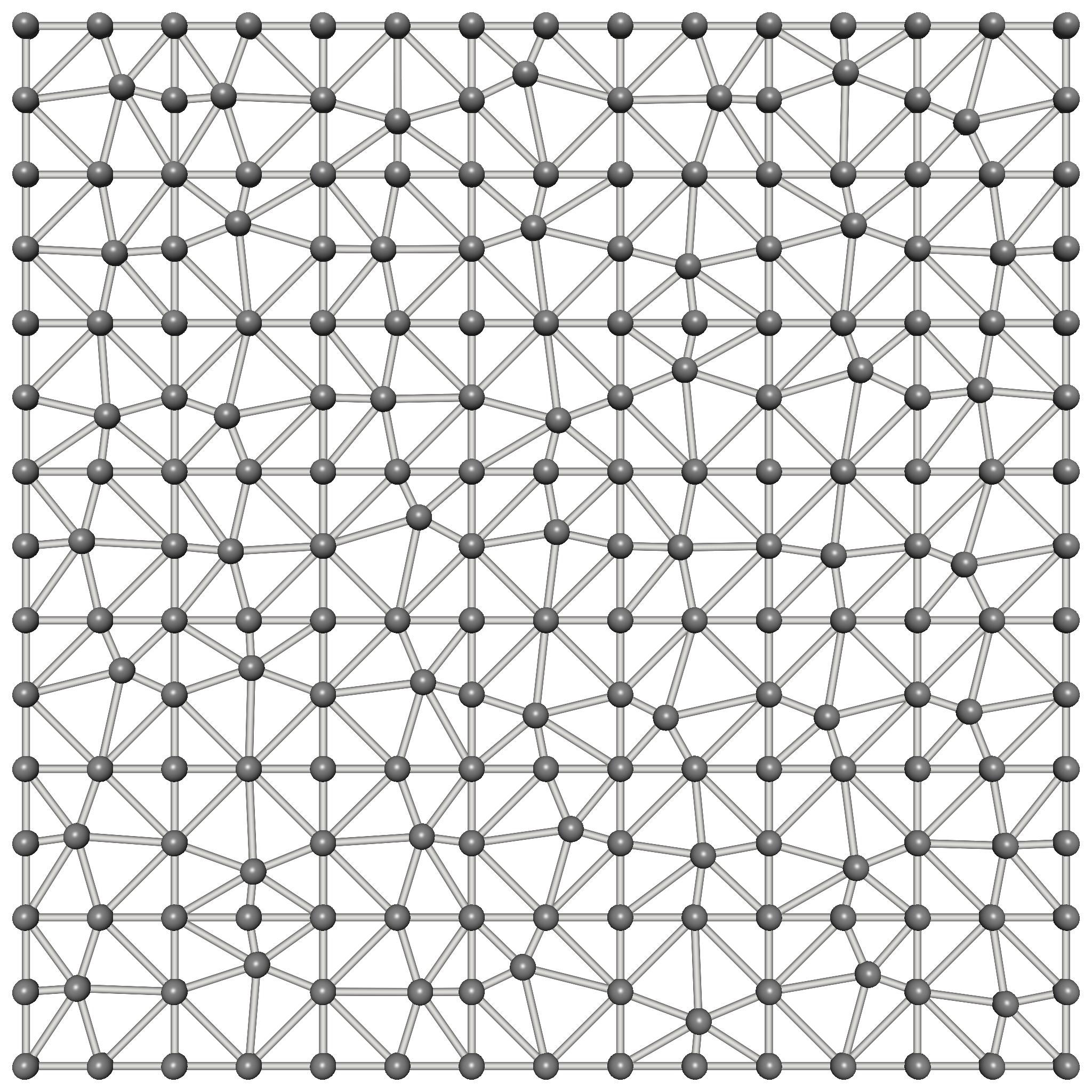}%
\caption{Triangle grid obtained by dividing each cell of a uniform grid into several triangles, including additional vertices obtained through bilinear interpolation.}%
\label{fig-ModifiedGrid}%
\end{subfigure}%
\caption{Different grid types as used in this paper.}%
\label{fig-AllGridTypes}%
\end{figure}

\pgfplotstableread{tikzImageData/uniformdirection.dat}		{\uniformdirection}
\pgfplotstableread{tikzImageData/uniformdegree.dat}   		{\uniformdegree}
\pgfplotstableread{tikzImageData/triangledirection.dat}		{\triangledirection}
\pgfplotstableread{tikzImageData/triangledegree.dat}   		{\triangledegree}
\pgfplotstableread{tikzImageData/poissondirection.dat}		{\poissondirection}
\pgfplotstableread{tikzImageData/poissondegree.dat}   		{\poissondegree}
\pgfplotstableread{tikzImageData/suggestiondirection.dat}	{\suggestiondirection}
\pgfplotstableread{tikzImageData/suggestiondegree.dat}   	{\suggestiondegree}

\newcommand{\StatPlot}[3]
{%
\tikzsetnextfilename{#1Direction}%
\begin{tikzpicture}[font=\scriptsize]%
\begin{axis}
[
	ybar,
	bar width=5pt,
	ylabel={number of edges},
	ymin=0,
	xtick={0, 45, 90, 135, 180},
	xlabel={grid edge direction},
	scale=0.45
]
\addplot[ybar, fill=blue!20] table[x=degree, y=y]{#2};
\end{axis}
\end{tikzpicture}%
\hfill%
\tikzsetnextfilename{#1Degree}%
\begin{tikzpicture}[font=\scriptsize]%
\begin{axis}
[
	ybar,
	bar width=7.5pt,
	ylabel={number of vertices},
	ymin=0,
	xtick={2, 3, ..., 10},
	xmax=11,
	xlabel={vertex degree},
	scale=0.45
]
\addplot[ybar, fill=blue!20] table[x=x, y=y]{#3};
\end{axis}
\end{tikzpicture}%
}

\begin{figure}[!t]
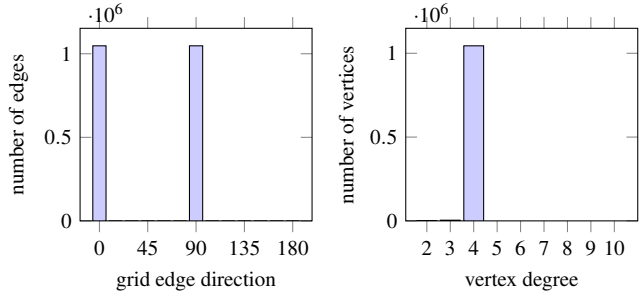
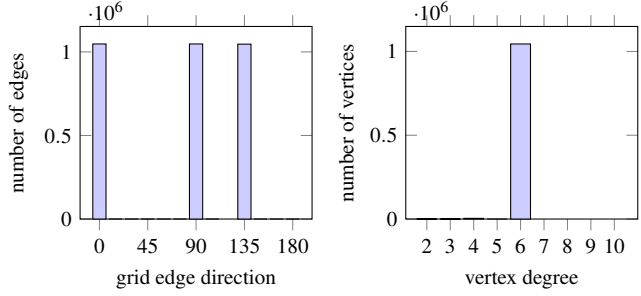
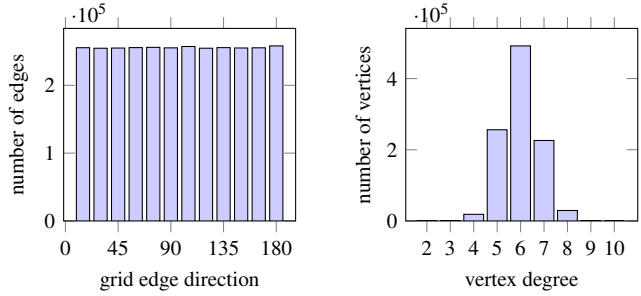
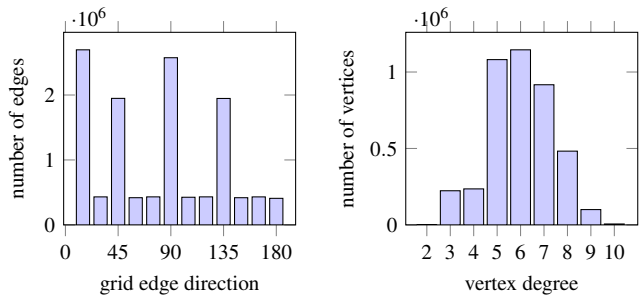
%
\begin{subfigure}{\linewidth}%
\StatPlot{UniformGridStats}{\uniformdirection}{\uniformdegree}%
\caption{Uniform grid.}%
\label{fig-UniformGridStats}%
\end{subfigure}%
\\[\baselineskip]%
\begin{subfigure}{\linewidth}%
\StatPlot{TriangleGridStats}{\triangledirection}{\triangledegree}%
\caption{Triangle grid obtained by diagonally dividing each cell of a uniform grid once as shown in \Cref{fig-InducedGrid}.}%
\label{fig-TriangleGridStats}%
\end{subfigure}%
\\[\baselineskip]%
\begin{subfigure}{\linewidth}%
\StatPlot{PoissonGridStats}{\poissondirection}{\poissondegree}%
\caption{Triangle grid with random vertex positions and Delaunay triangulation as shown in \Cref{fig-RandomGrid}.}%
\label{fig-PoissonGridStats}%
\end{subfigure}%
\\[\baselineskip]%
\begin{subfigure}{\linewidth}%
\StatPlot{SuggestionGridStats}{\suggestiondirection}{\suggestiondegree}%
\caption{Triangle grid derived from the uniform grid following our suggestion, as illustrated in \Cref{fig-ModifiedGrid} and discussed in \Cref{section-Suggestion}.}%
\label{fig-SuggestionGridStats}%
\end{subfigure}%
\caption{Statistics of different grids. Shown are the histogram of the grid edge directions (left) and the histogram of the vertex degrees (right).}
\label{fig-GridStats}
\end{figure}

In this section, we delve into the improvement in geometry that can be achieved by employing the steepest descent method of Robins et al. \cite{robins10}
on different grid types.
To this end,
we refer the reader
to \Cref{fig-AllGridTypes}
where we illustrate the types of grids that will be discussed
in this and the following section.

The pre-assumption of the probabilistic methods \cite{gyulassy12ms,reininghaus2012combinatorial} is that the data is given on a uniform grid. 
With a uniform grid, a cell has a limited number of higher dimensional cells to pair with during the calculation of the gradient vector field: 
a maximum of $4$ choices for the $2$D case, and a maximum of $6$ choices for the $3$D case. These numbers remain constant even when we increase the resolution of the grid.
Furthermore, due to the nature of the grid, the directions of the grid edges are either vertical or horizontal,
which limits the geometric expression of the separatrices drastically.
Hence, the steepest descent method \cite{robins10}
is not able to capture the geometry of the tangent curves on such grid.

It is interesting to investigate
the output of the steepest descent method
on other types of grids.
Specifically,
we want to investigate triangle grids
as they allow for a wider variety in the directions of the grid edges
as well as a higher degree of the vertices.

A straightforward approach is to divide each cell of a uniform grid into two triangles.
This fairly common choice of grid is illustrated in \Cref{fig-InducedGrid}.
It allows all non-boundary vertices
to have $6$ possible choices of direction
instead of just $4$.
Furthermore, it introduces diagonal grid edges.
\Cref{fig-TriangleGridStats}
shows the distributions of the vertex degrees
and of the grid edge directions
of a $1024 \times 1024$ grid of this kind.
The grid edge directions are measured
as the angle to the $x$-axis.
We used such a triangle grid
to sample the function of \Cref{formula-AnaFunc1}
and applied the steepest descent method.
The result is shown in \Cref{fig-AnaFunc1Triangular}.
We observe a slightly rounder shape for the circular part,
but overall the geometric embedding
is still quite different
compared to the continuous case from \Cref{fig-AnaFunc1Cont}.

Next,
we investigate
a triangle grid
with randomly placed vertices
as shown in \Cref{fig-RandomGrid}.
It is built
by uniformly placing vertices at the boundary
and then randomly placing the remaining vertices inside
using a Poisson disc sampling.
A Delaunay triangulation is then applied to obtain the mesh.
The statistics of this grid
are shown in \Cref{fig-PoissonGridStats}
and reveal
a high diversity
in terms of grid edge directions and vertex degrees.
Again,
we sampled \Cref{formula-AnaFunc1}
onto a $1024 \times 1024$ version of this grid
and computed the discrete gradient
using the steepest descent method.
The resulting Morse-Smale complex
(\Cref{fig-AnaFunc1RandomDelaunay})
comes remarkably close
to the continuous version
(\Cref{fig-AnaFunc1Cont}).
In fact,
in all our experiments,
this type of grid
exhibited
the most accurate geometry
of all grid types.

We repeated our experiments with another function $g$ sampled in the domain $[-2,2]^2$
\begin{align}
g(x,y) = \sin(xy) \cos(x+y).
\label{formula-AnaFunc2}
\end{align}
The results are shown in \Cref{fig-AnaFunc2Ms}.
Our previous observation
is confirmed:
the steepest descent method
does produce accurate geometry
on grids with high diversity
of the grid edge directions and vertex degrees.

We conclude that
when extracting the Morse-Smale complex in certain applications,
if possible,
a grid should be chosen
that allows for
a higher ``degree of freedom''
for the edge directions and vertex degrees.
The method by Robins et al. \cite{robins10}
will then be able to extract accurate geometry.

\section{Suggestion for Uniform Grids}
\label{section-Suggestion}

It is not always possible to choose a type of grid for an application.
Most applications come with pre-defined grids
and some of those will be uniform grids,
or very similar to uniform grids.
In this section, we propose a simple solution to 
achieve acceptable geometry for the Morse-Smale complex
given data sampled on a uniform grid
while 
preserving its topology through the steepest descent direction.

Our proposal draws inspiration from the observations in \Cref{section-EffectOfSampling}: 
better geometry can be achieved by allowing a higher number of options to choose from
during the computation of the discrete vector.
Thus, given a data set sampled on a uniform grid,
we create a triangle mesh
using well-chosen auxiliary vertices
to obtain the desired grid characteristics.

Our goal is to convert each uniform grid cell
into more than two triangles.
To do so,
we add one auxiliary vertex in the middle of each grid edge,
and another auxiliary vertex randomly into the interior of each grid cell.
The data values for these new vertices are obtained using bilinear interpolation.
A Delaunay triangulation within each grid cell then yields the final triangle mesh.
An example is shown in \Cref{fig-ModifiedGrid}.
Due to the random nature of the new points, 
spurious features can be created.
We apply an $\epsilon$-simplification
to cancel these unwanted features.
\changed{In our experiments, 
a persistence threshold of at most $0.001\%$ 
of the data range is applied.}

\Cref{fig-SuggestionGridStats}
shows the characteristics of this grid.
We can see that this grid covers all possible grid directions, albeit to a varying degree.
Some directions clearly dominate,
which stems from the fact
that this triangle grid
inherits the large number
of vertical and horizontal grid edges
from the original uniform grid.
Nevertheless,
the histograms show a vast improvement over the uniform grid.

We applied this procedure
to the uniform grids of size $1024 \times 1024$
that had been originally used to sample
\Cref{formula-AnaFunc1,formula-AnaFunc2}.
The resulting triangle meshes
have then be used to compute
the Morse-Smale complexes
via the steepest descent method.
The results are shown in
\Cref{fig-AnaFunc1MSSuggestion,fig-AnaFunc2MSSuggestion}, respectively.
Note how the geometric embedding
is much more aligned
with the continuous case
in comparison to the steepest descent method applied to the uniform grid.
In particular,
the circle in \Cref{fig-AnaFunc1MSSuggestion}
is well represented.
However, we can also see that
the diagonal separatrix in the lower-left corner
of the same figure
does not end at the corner of the domain,
but only nearby.

It remains to see
whether this new type of grid
gives rise to a Morse-Smale complex
that is equivalent to the one
obtained from the uniform grid
-- at least after an $\epsilon$-simplification.
We have reason to believe that this is true,
since we apply the steepest descent method
in both cases.
\changed{
Further, the original vertices are kept
in the new grid,
and the bilinearly interpolated values
of the new vertices
will always fall within 
the range of the original values.
Thus, the original minima remain the same, 
the saddle points and maxima are in the 
now refined lower star of the original vertex.
Although spurious minima, saddles, or maxima can 
be created due to the randomness of our suggestion,
they tend to be short-lived and can be canceled
at a very low persistence threshold as mentioned.
With more carefully chosen and optimized methods,
new points can be added without creating new criticals.
}

An experimental validation
is shown in \Cref{fig-Mat-sugg-cont},
where we show the Morse-Smale complexes
obtained using our suggested method
and using the continuous approach.
These should be compared also to
\Cref{fig-Mat1Interpolate}
where the steepest descent 
and the probabilistic method
have been applied to the original uniform grid.
We can clearly see
that our method avoids the topological issues of the probabilistic method,
while achieving a good geometric embedding.

\changed{
Our discussion in this section is limited to 
$2$D cases. Given a cubical grid, 
we suspect that a \emph{face-centered subdivision},
similar to the one by
Carr et al.\ \cite{carr2006}, could achieve similar 
results to ours. We left this for further investigations.
}

\added{
\paragraph{Increased computational effort}
Our suggestion leads to increased memory requirements and computation times.
Suppose that we add $a$ points at each edge and $b$ 
points inside each cell. 
For a grid of size $n \times n$, we added $2a(n-1)n$ 
points at the edges and $b(n-1)^2$ points in the cells,
a total of $(n-1)((2a+b)n - b)$ points in addition to the 
original $n^2$ points,
which is a substantial memory cost,
especially considering that
the newly created unstructured grid
requires explicit handling of the connectivity
whereas this could be handled implicitly for the original structured grid.
The computation time required to create 
such a triangulation is also non-negligible:
it takes in the order of tens of seconds for
larger grids such as $2048^2$.}

\added{Despite having higher computation time and memory requirements,
our method is able to provide Morse-Smale complexes
with accurate geometry and topology 
while being algorithmically less involved than the probabilistic variants.}

\begin{figure}%
\begin{subfigure}[t]{\ltwopicwidth}%
\includegraphics[width=\linewidth]{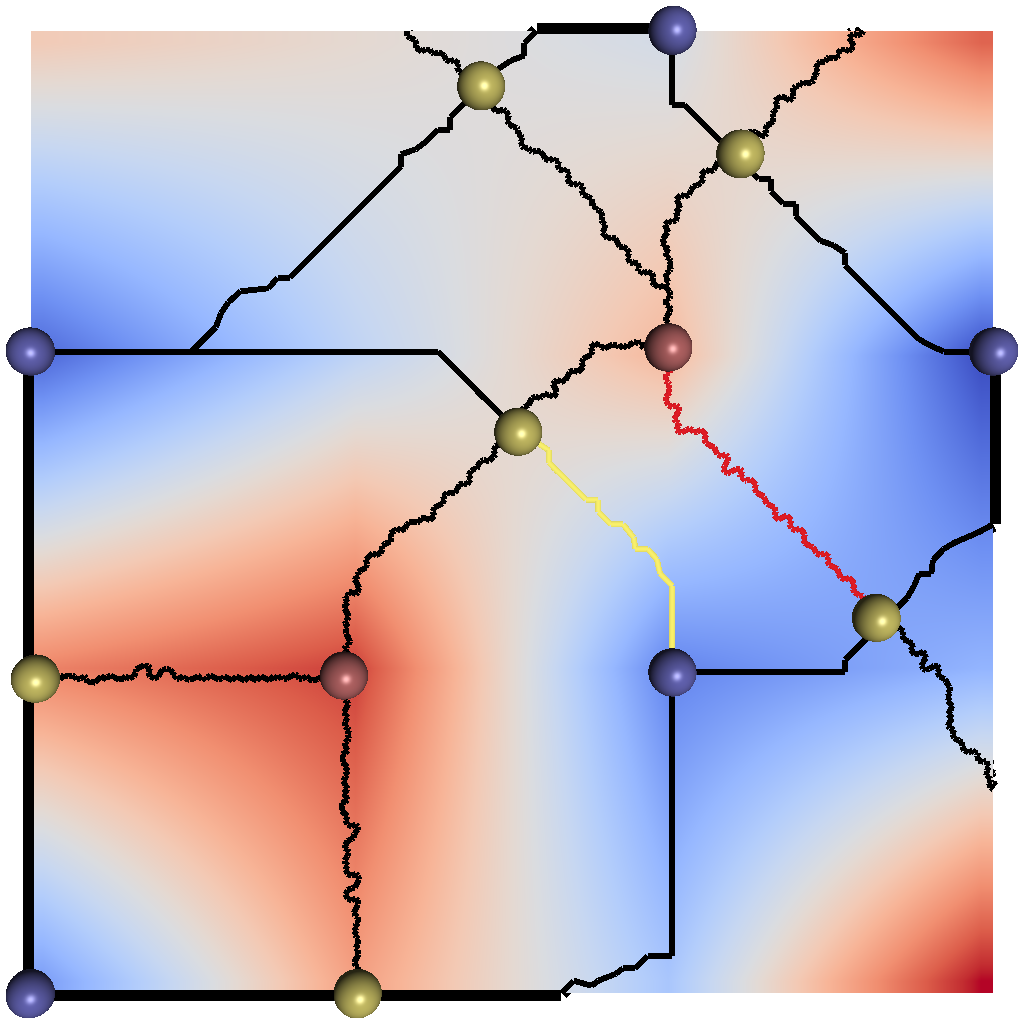}
\caption{Steepest descent method applied to our suggested triangle grid.}%
\label{fig-Mat-sugg}%
\end{subfigure}
\hfill%
\begin{subfigure}[t]{\ltwopicwidth}%
\includegraphics[width=\linewidth]{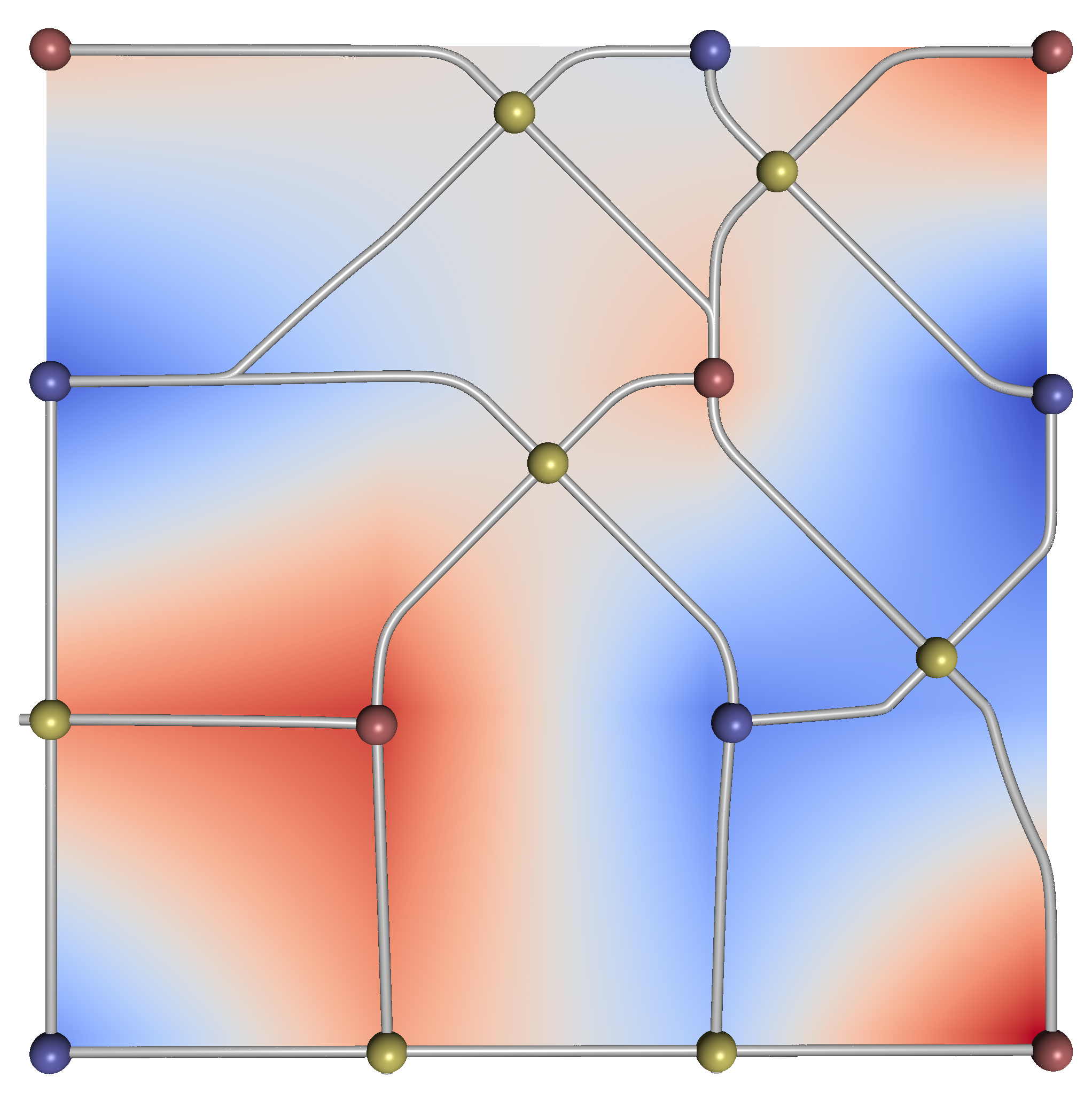}
\caption{Continuous Morse-Smale complex.}%
\label{fig-Mat-cont}%
\end{subfigure}
\caption{The discrete Morse-Smale complex obtained using the steepest descent method applied to our suggested triangle grid in comparison to the continuous Morse-Smale complex. The latter has a slightly different handling of the cases at the boundary, but otherwise the topologies coincide. Compare these results to \Cref{fig-Mat1Interpolate}.}%
\label{fig-Mat-sugg-cont}%
\end{figure}

\newlength{\anafunctwoimgwidth}
\setlength{\anafunctwoimgwidth}{0.32\textwidth}

\begin{figure*}[t]%
\begin{subfigure}[t]{0.985\anafunctwoimgwidth}%
\includegraphics[width=\linewidth]{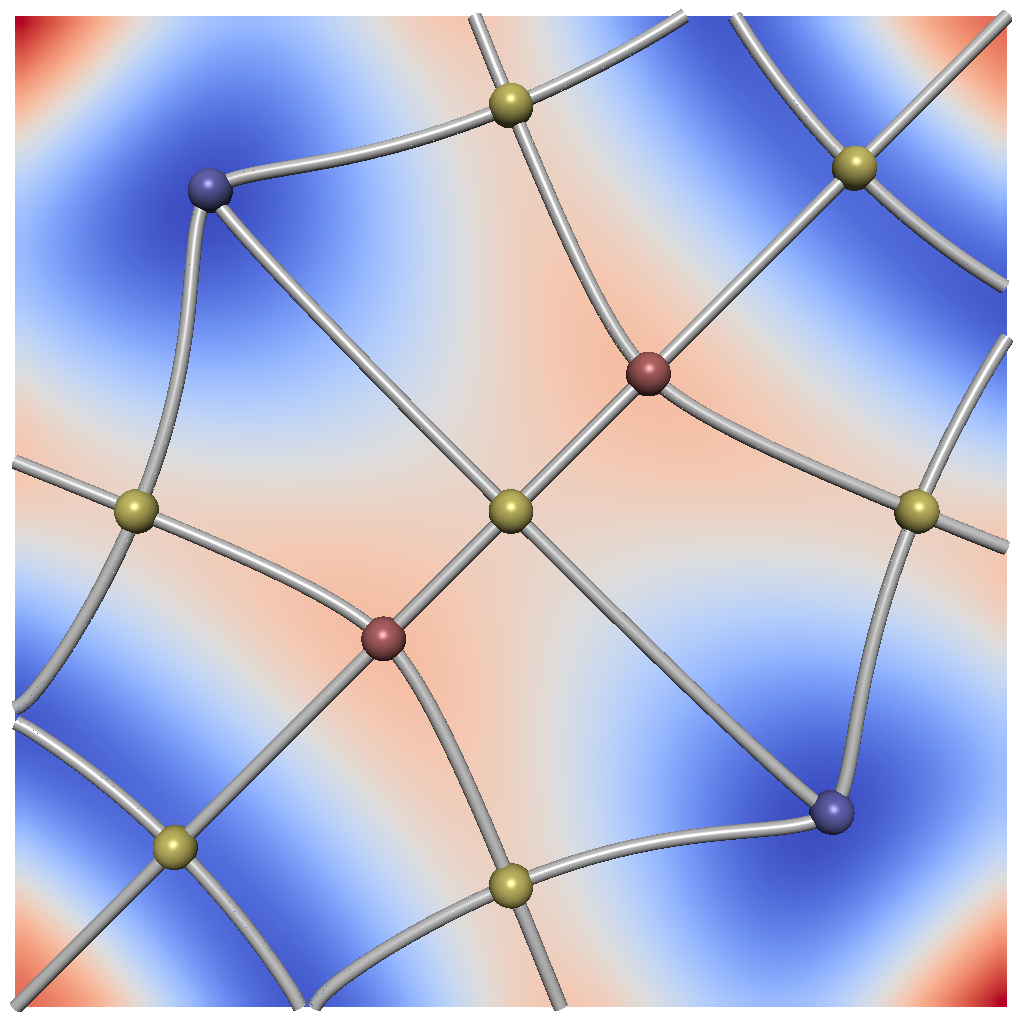}%
\caption{Continuous topology.}%
\label{fig-AnaFunc2Cont}%
\end{subfigure}%
\hfill%
\begin{subfigure}[t]{\anafunctwoimgwidth}%
\includegraphics[width=\linewidth]{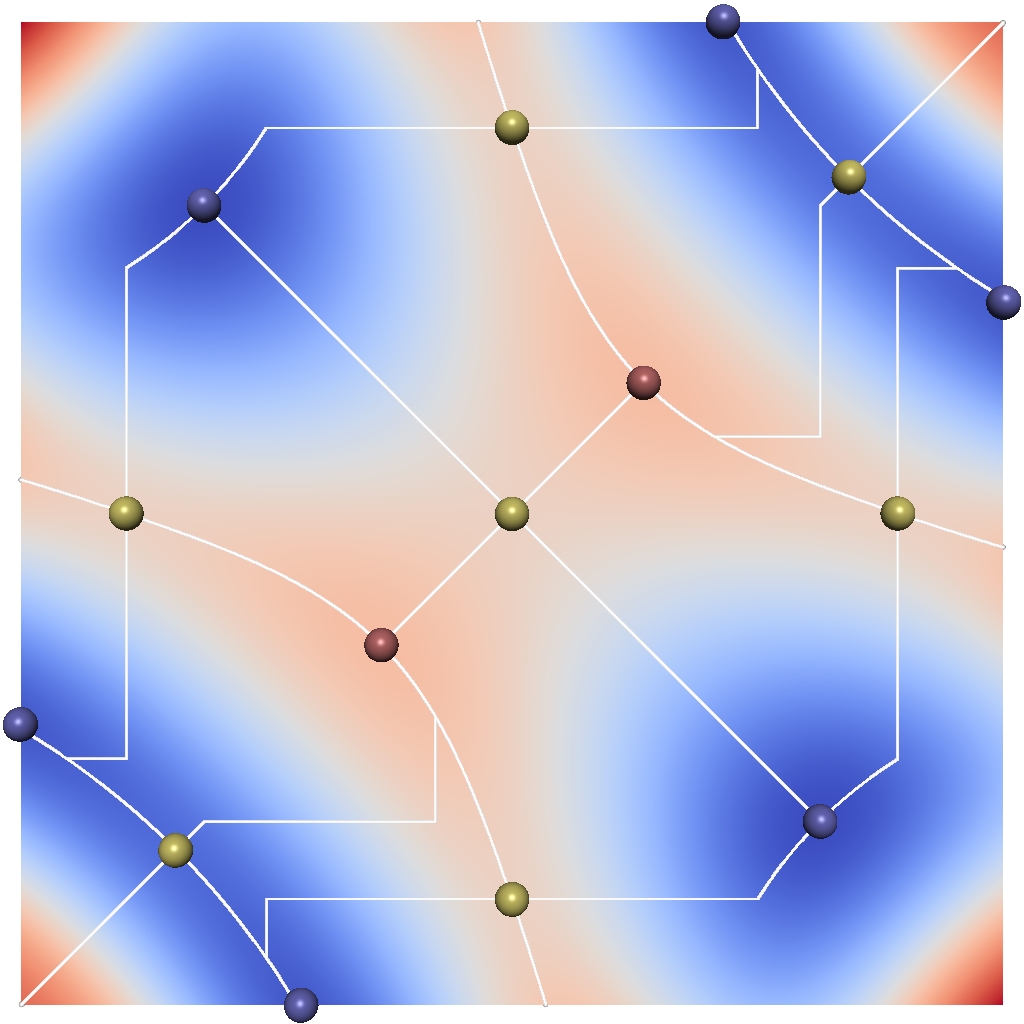}%
\caption{Steepest descent method \cite{robins10} on the uniform grid.}%
\label{fig-AnaFunc2Robins}%
\end{subfigure}%
\hfill%
\begin{subfigure}[t]{\anafunctwoimgwidth}%
\includegraphics[width=\linewidth]{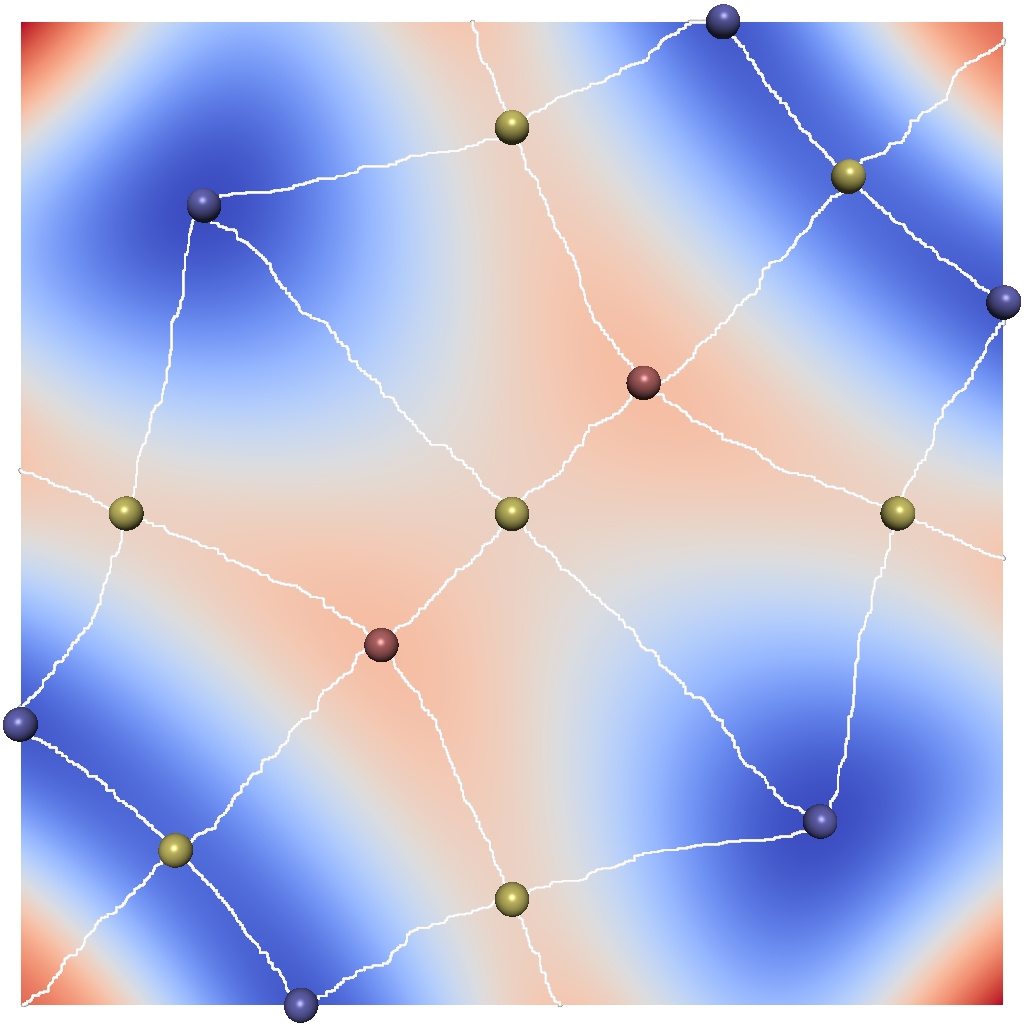}%
\caption{Probabilistic method \cite{gyulassy12ms, reininghaus2012combinatorial} on the uniform grid.}%
\label{fig-AnaFunc2Jans}%
\end{subfigure}%
\\%
\begin{subfigure}[t]{0.985\anafunctwoimgwidth}%
\includegraphics[width=\linewidth]{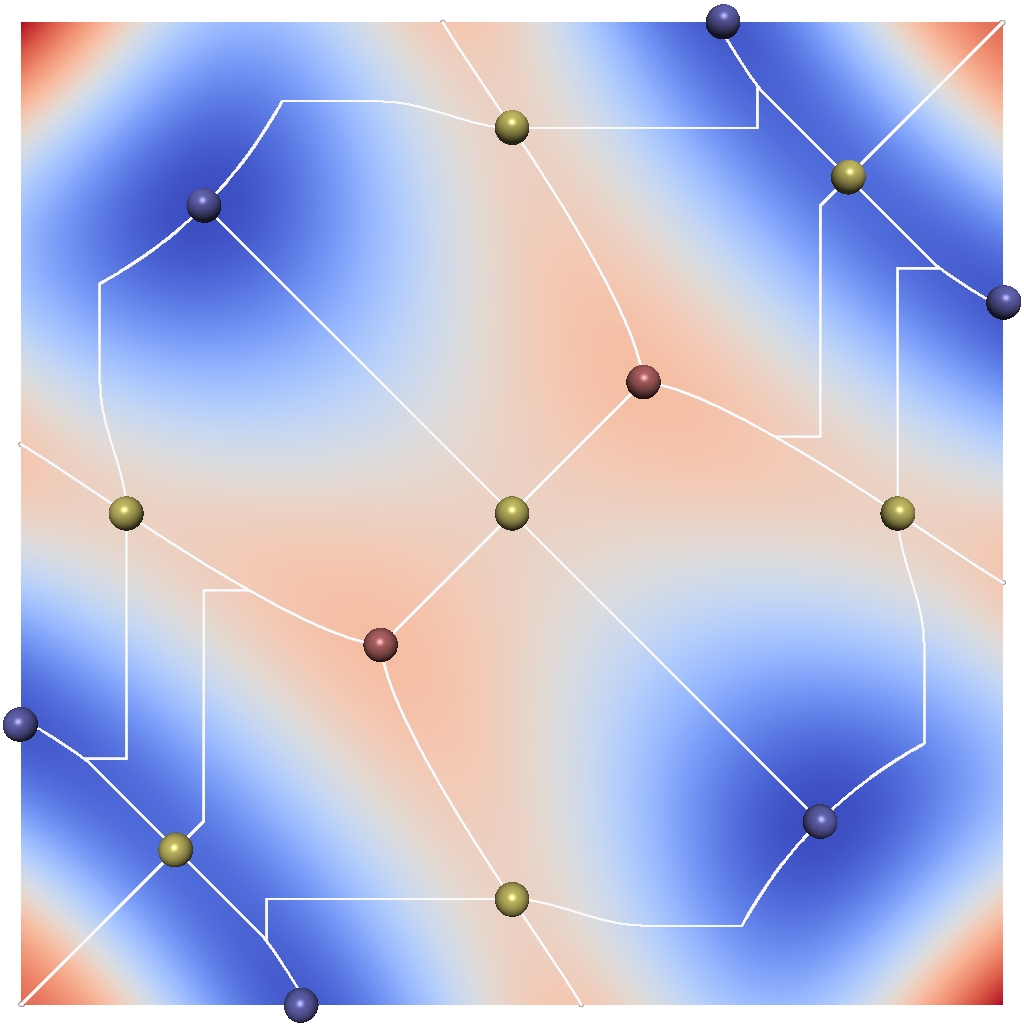}%
\caption{Steepest descent method \cite{robins10} on a triangle grid obtained by diagonally dividing each cell of a uniform grid once as shown in \Cref{fig-InducedGrid}.}%
\label{fig-AnaFunc2Triangular}%
\end{subfigure}%
\hfill%
\begin{subfigure}[t]{\anafunctwoimgwidth}%
\includegraphics[width=\linewidth]{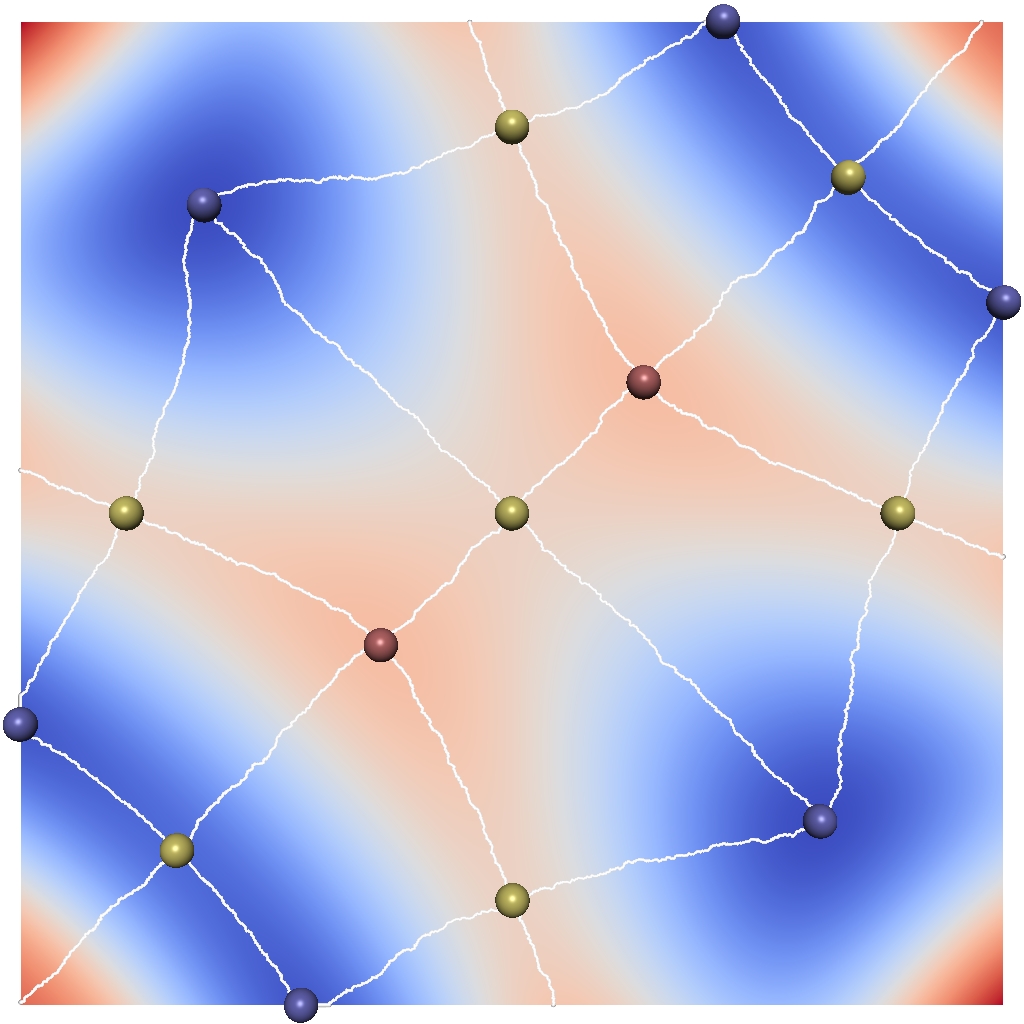}%
\caption{Steepest descent method \cite{robins10} on a triangle grid with random vertex positions and Delaunay triangulation as shown in \Cref{fig-RandomGrid}.}%
\label{fig-AnaFunc2RandomDelaunay}%
\end{subfigure}%
\hfill%
\begin{subfigure}[t]{\anafunctwoimgwidth}%
\includegraphics[width=\linewidth]{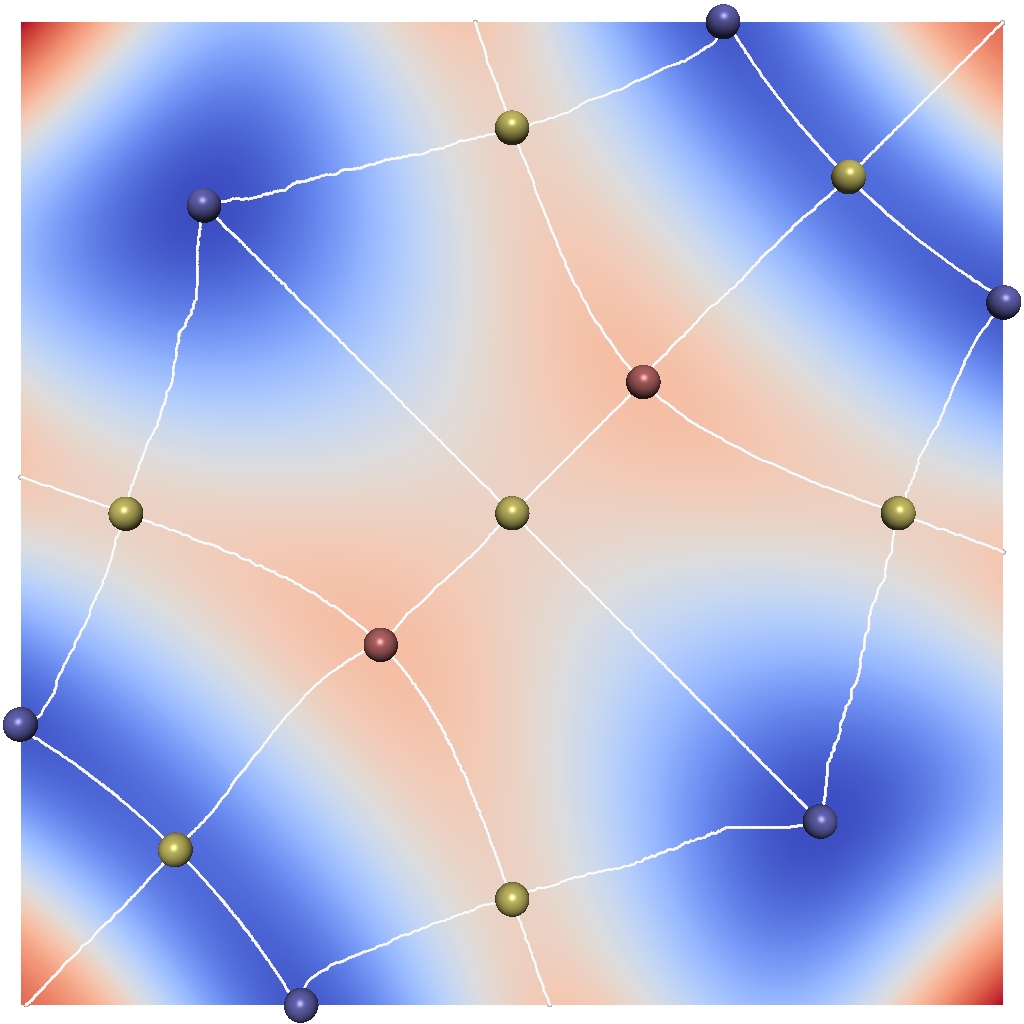}%
\caption{Steepest descent method \cite{robins10} on a triangle grid derived from the uniform grid following our suggestion, as illustrated in \Cref{fig-ModifiedGrid} and discussed in \Cref{section-Suggestion}.}%
\label{fig-AnaFunc2MSSuggestion}%
\end{subfigure}%
\caption{The function from \Cref{formula-AnaFunc2} has been sampled on different types of grids and the Morse-Smale complex has been extracted using different methods. Qualitatively, we see the same results here as in \Cref{fig-AnaFunc1MS}, namely that \emph{accurate geometry} can be achieved with different methods, most notably also with the steepest descent method if the right kind of grid is provided.}%
\label{fig-AnaFunc2Ms}%
\end{figure*}

\section{Conclusion}
\label{section-outro}

We discussed the different aspects that can affect the quality of the Morse-Smale complex in terms of geometry and topology. 
We have shown that the undesired shape of the Morse-Smale complex using the traditional methods comes from the low number of choices for the 
vector paring during the computation. By allowing a higher number of options for the vectors through a more optimized grid structure,  
we can attain a Morse-Smale complex that closely resembles its continuous version.
It is thus encouraged, if possible, to take this into account when sampling the data.

However, not all applications allow such sampling methods
and we need to employ methods for accurate geometry.
While it is known that the existing
methods produce geometrically different embeddings for the separatrices,
we have shown that there are also unintended consequences in terms of 
topological changes.
We have shown that the $1$-saddles can take different positions,
which can lead to significant alterations 
in the overall structure of the Morse-Smale complex.
However,
we did not investigate all methods that fall under the category of \emph{accurate geometry},
e.g., the work by Gyulassy et al.\ \cite{gyulassy19ms}.
We leave this to future work.

We also proposed a method to achieve better geometry for the Morse-Smale complex given data sampled on a uniform grid, drawing inspiration from our observations. 
Compared to the other methods aimed at tackling the geometry problem, our suggestion is significantly less complicated. This method required no modification 
on the algorithm side, as we only added more information to the given data. Therefore, our suggestion does not increase the complexity of these methods,
albeit it does increase the memory requirements
\changed{and computation times}.
Despite its simplicity, our method produced comparable results in terms of geometry while retaining the topological features inherited from the steepest descent method.

\acknowledgments{
This work was supported through grants from The Swedish Research Council (Vetenskapsrådet, project 2020-05461), and the Swedish e-Science Research Centre (SeRC).
The visualizations were created using the Inviwo framework \cite{inviwo2019}.
The Hurricane Isabel data set was produced by the Weather Research and Forecast (WRF) model, courtesy of NCAR, and the U.S. National Science Foundation (NSF).}

\bibliographystyle{abbrv-doi}

\bibliography{Literature}
\end{document}